\pdfoutput=1
\documentclass[12pt]{article}%
\usepackage{amsmath}
\usepackage{amsthm}
\usepackage{graphicx}
\usepackage{color}
\usepackage{url}%
\usepackage{amsfonts}%
\usepackage{amssymb}
\newtheorem{theorem}{Theorem}

\newtheorem{corollary}[theorem]{Corollary}

\begin{document}

\title{On estimating extremal dependence structures by parametric spectral measures}
\author{Jan Beran$^{1}$ and Georg Mainik$^{2}$\\ 
\\
$^{1}$Department of Mathematics and Statistics,\\
University of Konstanz\\
$^{2}$RiskLab, Department of Mathematics, ETH Zurich}
\maketitle

\begin{abstract}
Estimation of extreme value copulas is often required in situations
where available data are sparse. Parametric methods may then be the 
preferred approach. 
A possible way of defining parametric families that are simple and,
at the same time, cover a large variety of multivariate
extremal dependence structures is to 
build models based on spectral measures.
This approach is considered here. 
Parametric families of spectral measures   
are defined as convex hulls of suitable basis elements, 
%Convex combinations of suitable basis elements 
%suitable spectral basis functions 
and parameters 
are estimated by projecting an initial nonparametric estimator 
on these finite-dimensional spaces.  
Asymptotic distributions are derived for the estimated parameters 
and the resulting estimates of the spectral measure and the extreme 
value copula.
% The asymptotic distribution of estimated parameters and 
% resulting estimates of the spectral measure and the extreme value 
% copula are derived. 
Finite sample properties are illustrated by a simulation study.
\end{abstract}

%\bibliographystyle{plain}
%this sets the citation/bibliography style
%
%
%%%%   Citation examples for the natbib package:   %%%% 
%
% Compatibility with standard commands:
%\cite{Gudendorf/Segers:2011, Fils-Villetard/Guillou/Segers:2008}.
%
% only the year in brackets:
% \citet{Gudendorf/Segers:2011, Fils-Villetard/Guillou/Segers:2008}.
%
% All citation in brackets, no additional brackets for the year:
% \citep[cf.][and references therein]{Gudendorf/Segers:2011, 
%Fils-Villetard/Guillou/Segers:2008}.
%
% Write out all authors:
%\citet*{Gudendorf/Segers:2011, Fils-Villetard/Guillou/Segers:2008}.
%
%
%%%%   Margin notes:    %%%%%
%
%\mynote{Add a comment here (gm)}

\section{Introduction}

Extreme value copulas provide a suitable general approach to modelling
multivariate extremes. Various nonparametric methods for estimating extreme
value copulas have been proposed in the last few years
\cite{Zhang/Wells/Peng:2008}, \cite{Gudendorf/Segers:2011}, 
\cite{Gudendorf/Segers:2012},
\cite{Caperaa/Fougeres/Genest:1997} (also see \cite{Pickands:1981},
\cite{Deheuvels:1991} and \cite{Hall/Tajvidi:2000} for related approaches). In
practical applications, such as for instance operational risk or rare natural
disasters, one is however often in a situation where available data are
sparse. Nonparametric methods generally require a fairly large sample size in
order to be reliable. For small samples and in situations where one may have
some idea about plausible properties of the distribution, parametric methods
are likely to yield more accurate results. An approach to parametric inference
for extreme value copulas is discussed for instance in
\cite{Boldi/Davison:2007}.

One of the key issues is how to define parametric families that are simple and
at the same time general enough to cover a large variety of multivariate
dependence structures in the extremes. For instance, some of the most popular
models are based on Archimedean copulas, which all correspond to the same type
of extremal dependence structure, characterized by the Gumbel copula
\cite{Genest/Rivest:1989}. One way of achieving more 
flexibility in the extremes is to build models based on spectral measures.
This is the approach taken here. For related work see e.g.
\cite{Fils-Villetard/Guillou/Segers:2008}, 
\cite{Guillotte/Perron/Segers:2011}, 
and \cite{Gudendorf/Segers:2012}. 

More specifically, the idea pursued in the following is to select a finite 
number of suitable spectral measures as basis elements and to use their convex 
combinations as a parametric family of dependence structures. 
Given a sufficiently large
number of such basis elements, any spectral measure can be approximated by a
weighted sum. 
Estimation of the coefficients can then be carried out by
projecting a nonparametric estimator, such as the one in
\cite{Caperaa/Fougeres/Genest:1997}, on the finite-dimensional space 
generated by the basis elements.
If the number of basis elements in the 
model is large (and increasing with the sample size), then
projecting the original non-parametric estimator can 
be considered as a discretization technique. 
This is the setting in \cite{Fils-Villetard/Guillou/Segers:2008} 
and \cite{Gudendorf/Segers:2012}. 

On the other hand, an appropriate model with a small number of basis elements 
can have the advantage of dimension reduction. 
Given a reasonable parametric model with a small number of parameters, 
one can reduce the variability of a nonparametric estimator by projecting it 
on a low-dimensional space. This is the approach studied here. 
We define explicit parameter estimators in the low-dimensional setting 
and study the asymptotic distribution of the resulting estimators 
of the dependence structure. 
To illustrate the potential advantage of dimension reduction,
we construct an example with three basis elements and compare a 
non-parametric estimator with its low-dimensional projection in a simulation
study.

Note that in principle any nonparametric estimator 
(cf.   
\cite{Deheuvels:1991}, \cite{Caperaa/Fougeres/Genest:1997},
\cite{Hall/Tajvidi:2000}, \cite{Zhang/Wells/Peng:2008},
\cite{Gudendorf/Segers:2011}, \cite{Gudendorf/Segers:2012})
can be used as a starting point.
Depending on the nonparametric method 
used in the projection, the marginal distributions are either 
known or estimated from the observed data. 
The asymptotic results given below only require that a functional 
limit theorem in a suitable topology 
holds for the initial estimator.

The paper is organized as follows. Basic definitions and concepts of
multivariate extreme value theory are summarized in section 2. Parametric
models in the spectral domain and a corresponding parametric estimator are
introduced in section 3. Asymptotic results, including consistency and a
central limit theorem, are derived in section 4. 
The theoretical results are illustrated by simulations for a
specific model in section 5. Final remarks in section 6 
with a discussion of some open problems conclude the paper.

\section{Basic definitions}

Consider a sample $\mathbf{X}_{1},\ldots,\mathbf{X}_{n}$ consisting of iid
realizations $\mathbf{X}_{i}=\left(  X_{i,1},\ldots,X_{i,d}\right)  ^{T}$ of a
$d$-dimensional random vector $\mathbf{X}=\left(X_{1},\ldots,X_{d}\right)^{T}\in\mathbb{R}^{d}$ with marginal
distributions $F_{1},\ldots,F_{d}$ and copula $C_{\mathbf{X}}$. That is,
$F_{j}(t)=P\left(  X_{j}\leq t\right)  $ for $j=1,\ldots,d$ and $t\in
\mathbb{R}$, and
\[
P\left(  \mathbf{X}\leq\mathbf{x}\right)  =C_{\mathbf{X}}\left(  F_{1}\left(
x_{1}\right)  ,\ldots,F_{d}\left(  x_{d}\right)  \right)
\]
for $\mathbf{x}\in\mathbb{R}^{d}$. The notation $\mathbf{x}\leq\mathbf{y}$ for
$\mathbf{x},\mathbf{y}\in\mathbb{R}^{d}$ means $x_{j}\leq y_{j}$ for
$j=1,\ldots,d$. The transposition operator $(\cdot)^{T}$ in $\mathbf{X}=\left(X_{1},\ldots,X_{d}\right)^{T}$ indicates that $\mathbf{X}$ is considered as a column vector. Distinguishing columns and rows will be useful in some calculations 
later on. 

The vector $\mathbf{M}_{n}=\left(  M_{n,1},\ldots
,M_{n,d}\right)  ^{T}$ of componentwise maxima%
\[
M_{n,j}=\max_{i=1,2,\ldots,n}X_{i,j}%
\]
then has marginal distributions $P(M_{n,j}\leq t)=F_{j}^{n}(t)$ and a copula
$C_{\mathbf{M}_{n}}\left(  u\right)  $ given by
\[
P\left(  \mathbf{M}_{n}\leq\mathbf{x}\right)  =C_{\mathbf{M}_{n}}\left(
F_{1}^{n}\left(  x_{1}\right)  ,\ldots,F_{d}^{n}\left(  x_{d}\right)  \right)
=C_{\mathbf{X}}^{n}\left(  F_{1}\left(  x_{1}\right)  ,\ldots,F_{d}\left(
x_{d}\right)  \right)  .
\]
In the limit one obtains, under general conditions, an \emph{extreme value
copula} $C\left(  \mathbf{u}\right)  =\lim_{n\rightarrow\infty}C_{\mathbf{X}%
}^{n}\left(  \mathbf{u}^{1/n}\right)  $ ($\mathbf{u}=\left(  u_{1}%
,\ldots,u_{d}\right)  ^{T}\in\left[  0,1\right]  ^{d}$) with the
characteristic max-stable property
\begin{equation}
C\left(  \mathbf{u}\right)  =C^{n}\left(  \mathbf{u}^{1/n}\right)
\label{Eq:Max-Stable}%
\end{equation}
for all $n\in\mathbb{N}$. For an accessible introduction to this topic see
e.g. \cite{Gudendorf/Segers:2010} and references therein. The definition
(\ref{Eq:Max-Stable}) of extreme value or max-stable copulas is equivalent to
the representation
\[
C\left(  \mathbf{u}\right)  =\exp\left(  -\ell\left(  -\log u_{1},\ldots,-\log
u_{d}\right)  \right)
\]
with the \emph{tail dependence function}
\begin{equation}
\ell(\mathbf{x})=\int_{\Delta_{d}}\max_{i=1,\ldots,d}\left(  w_{i}%
x_{i}\right)  d\Psi\left(  w_{1},\ldots,w_{d}\right)  ,\quad(\mathbf{x}%
\in\lbrack0,\infty)^{d}), \label{eq:tail_dep_func}%
\end{equation}
and $\Psi$ the so-called \emph{spectral measure} on the unit simplex in
$\mathbb{R}^{d}$, $\Delta_{d}=\{\mathbf{x}\in\lbrack0,1]^{d}:\sum_{i=1}%
^{d}x_{i}=1\}$, satisfying
\begin{equation}
\int_{\Delta_{d}}w_{i}d\Psi(w_{1},\ldots,w_{d})=1,\quad(i=1,\ldots,d)
\label{Eq:Spectral_Measure_Condition}%
\end{equation}
(cf. Theorem 6.2.2 in \cite{Gudendorf/Segers:2010}). Note that the last
condition implies $\int_{\Delta_{d}}d\Psi\left(  \mathbf{w}\right)  =d$. The
underlying original results go back to \cite{de_Haan/Resnick:1977} and
\cite{Pickands:1981}. The main conclusion is that spectral measures, tail
dependence functions, and extreme value copulas are equivalent representations
of dependence structures in multivariate extreme value theory. Note also that
$\ell(r\mathbf{x})=r\ell(\mathbf{x})$ for $r>0$ and $\mathbf{x}\in
\lbrack0,\infty)^{d}$, so that it is sufficient to specify $\ell(\mathbf{x})$
for $\mathbf{x}\in\Delta_{d}$ only. The restriction of $\ell$ to $\Delta_{d}$
is also called Pickands dependence function and is usually denoted by
$A\left(  \cdot\right)  $. The extension to $\mathbf{x}\in\lbrack0,\infty
)^{d}$ is obtained by
\begin{equation} \label{eq:l_and_A}
\ell(\mathbf{x})=\Vert\mathbf{x}\Vert_{1}\ell\left(  \frac{\mathbf{x}}%
{\Vert\mathbf{x}\Vert_{1}}\right)  =\Vert\mathbf{x}\Vert_{1}A\left(
\mathbf{w}\right)
\end{equation}
where $\Vert\mathbf{x}\Vert_{1}=x_{1}+\cdots+x_{d}$ (since $x_{j}\geq0$), and
$\mathbf{w}=$ $\mathbf{x}/\Vert\mathbf{x}\Vert_{1}$. Note that condition~\eqref{Eq:Spectral_Measure_Condition} is equivalent to 
% means the following standardization of $A$ 
% on the vertices $\mathbf{e}_i$ of the unit simplex $\Delta_d$:
% equivalent to  
\begin{equation}\label{Eq:Spectral_Measure_Condition_A}
A(\mathbf{e}_{j}) = 1,\quad (j=1,\ldots,d)
\end{equation}
where $\mathbf{e}_{j}=(e_{j,1},...,e_{j,d})^{T}$ is the $j$-th unit vector in 
$\mathbb{R}^d$: $e_{j,l}=0$ ($j\neq l$) and $e_{j,j}=1$.
That is, \eqref{Eq:Spectral_Measure_Condition} standardizes $A$ on the vertices of the unit simplex $\Delta_d$. 

%\section{Spectral parametric models and parametric estimation}
\section{Parametric models for spectral measures: construction and estimation}

\subsection{Models}

One way of building parametric models that encompass a large variety of
extremal dependence structures is to start at the level of the spectral
measure $\Psi$. 
%We will call such families of extremal copulas \textit{spectral parametric models}. 
Thus, suppose that $\Psi_{1}\ldots,\Psi_{p}$ are some fixed spectral 
measures. The correspoding dependence
functions and extreme value copulas will be denoted by $\ell_{1},\ldots
,\ell_{p}$, $A_{1},\ldots,A_{p}$ and $C_{1},\ldots,C_{p}$, respectively. A
parametric family of spectral measures $\mathcal{P}_{p}=\left\{
\Psi\left(  \cdot,\mathbf{\theta}\right)  ,\mathbf{\theta}\in\Theta\right\}
$, and corresponding families $\mathcal{A}_{p}$ and $\mathcal{C}_{p}$ of
(Pickands) dependence functions and copulas respectively, can then be obtained
by defining spectral measures of the form%
\[
\Psi\left(  \cdot,\mathbf{\theta}\right)  =\sum_{i=1}^{p-1}\theta_{i}\Psi
_{i}(\cdot)+\left(  1-\sum_{i=1}^{p-1}\theta_{i}\right)  \Psi_{p}(\cdot)
\]
where $\mathbf{\theta}=\left(  \theta_{1},\ldots,\theta_{p-1}\right)
\in\Theta$ and $\Theta=\{\mathbf{\vartheta}\in(0,1)^{p-1}:\sum_{i=1}%
^{p-1}\vartheta_{i}\leq1\}$. As (\ref{Eq:Spectral_Measure_Condition}) remains
valid for convex combinations, $\Psi\left(  \cdot,\mathbf{\theta}\right)  $ is
a spectral measure by definition. In terms of the corresponding dependence
functions we have%
\begin{align*}
\ell\left(  \mathbf{x},\mathbf{\theta}\right)   &  =\sum_{i=1}^{p-1}\theta
_{i}\ell_{i}\left(  \mathbf{x}\right)  +\left(  1-\sum_{i=1}^{p-1}\theta
_{i}\right)  \ell_{p}\left(  \mathbf{x}\right)  ,\\
A\left(  \mathbf{w},\mathbf{\theta}\right)   &  =\sum_{i=1}^{p-1}\theta
_{i}A_{i}\left(  \mathbf{w}\right)  +\left(  1-\sum_{i=1}^{p-1}\theta
_{i}\right)  A_{p}\left(  \mathbf{w}\right)  .
\end{align*}
For the copulas we obtain
\begin{align}
C\left(  \mathbf{u},\mathbf{\theta}\right)   &  =\exp\left\{  -\sum
_{i=1}^{p-1}\theta_{i}\left[  \ell_{i}\left(  -\log\mathbf{u}\right)
-\ell_{p}\left(  -\log\mathbf{u}\right)  \right]  -\ell_{p}\left(
-\log\mathbf{u}\right)  \right\}  \label{eq:C(.,theta)}\\
&  = C_{p}\left(  \mathbf{u}\right)  \left(  \prod_{i=1}^{p-1}\left(
\frac{C_{i}\left(  \mathbf{u}\right)  }{C_{p}\left(  \mathbf{u}\right)
}\right)  ^{\theta_{i}}\right) \nonumber%
\end{align}
where $\ell_{i}\left(  -\log\mathbf{u}\right)  =\ell_{i}\left(  -\log
u_{1},\ldots,-\log u_{d}\right)$.

Henceforth we assume that the parameter $\theta$ is identifiable in 
the sense that $A_{\mathbf{\theta}}= A_{\mathbf{\theta}'}$ implies
$\mathbf{\theta}=\mathbf{\theta}'$ for 
$\mathbf{\theta},\mathbf{\theta}'\in\Theta$. 
A sufficient criterion for the identifiability of $\mathbf{\theta}$ 
is linear independence of the basis elements $A_1,\ldots,A_p$. 
If $\theta$ is not identifiable, then an estimator 
$\hat{\mathbf{\theta}}$ may fail to converge. 
An important example of this issue 
is the decomposition of discrete spectral measures.  
According to \cite{Mai/Scherer:2011}, any discrete spectral measure 
on $\Delta_2$ can be expressed as a convex combination of two-point 
spectral measures. 
This result can also be written in terms of piecewise linear 
dependence functions and Marshall-Olkin copulas. 
However, the decomposition is not necessarily unique. 
This can be illustrated by the following example. 
Let $(t,1-t)\in\Delta_2$ be represented by the first coordinate 
$t\in[0,1]$ and 
consider the family $\mathcal{Q}_4$ of discrete spectral measures 
$\Psi=\sum_{i=1}^{4}c_i\delta_{(i-1)/3}$ with $c_i\ge0$ 
for $i=1,\ldots,4$. 
It is easy to see that a basis of $2$-point spectral measures needed for the 
decomposition of all $\Psi\in\mathcal{Q}_4$ must include all elements of 
$\mathcal{Q}_4$ with only two atoms. These are
\[
\Psi_1=\delta_{0}+\delta_{1},\ 
\Psi_2=\delta_{1/3}+\delta_{2/3},\ 
\Psi_3=\frac{1}{2}\delta_{0} + \frac{3}{2}\delta_{2/3},\ 
\Psi_4=\frac{3}{2}\delta_{1/3} + \frac{1}{2}\delta_{1}.
\]
The non-uniqueness follows from 
$\frac{1}{4}\Psi_1 + \frac{3}{4}\Psi_2 = \frac{1}{2}(\Psi_3+\Psi_4)$. 

\subsection{Estimation}

Several nonparametric estimators of Pickands dependence functions, 
spectral measures, and corresponding extreme value copulas  
have been proposed in the recent literature
\cite{Deheuvels:1991}, \cite{Caperaa/Fougeres/Genest:1997},
\cite{Hall/Tajvidi:2000}, \cite{Zhang/Wells/Peng:2008},
\cite{Gudendorf/Segers:2011}, 
\cite{Gudendorf/Segers:2012}. 
Generally, these methods require fairly large
sample sizes in order to achieve a sufficient degree of accuracy. In contrast,
parametric estimates are expected to be reasonably accurate for moderate or
even small sample sizes, provided that the parametric assumptions are
sufficiently realistic. To see how much may be gained by parametric
estimation, we consider the following approach. Suppose that a nonparametric
estimate $\widehat{A}$ of the dependence function $A$ is given, and recall that the corresponding $\widehat{l}$ is obtained from $\widehat{A}$ according to~\eqref{eq:l_and_A}. 
A natural
parametric estimator based on the family $\mathcal{A}_{p}$, and
$\widehat{A}$ as intial estimate, is obtained by projecting the function
$\widehat{A}$ on $\mathcal{A}_{p}$. 
Note that even if \eqref{eq:tail_dep_func} does not hold for $\widehat{A}$ 
(see e.g. \cite{Gudendorf/Segers:2011}), it holds automatically for the projection of $\widehat{A}$ on $\mathcal{A}_{p}$, so that 
this projection is a proper dependence function by definition. 
Note also
that the projection improves the accuracy of the estimate if the true
dependence function is indeed in $\mathcal{A}_{p}$. 
Related improvements for projections on infinite-dimensional spaces of 
spectral measures and approximations by sieve methods have been considered in 
% Related {\it nonparametric}
% improvements have been considered in the literature for instance by
\cite{Fils-Villetard/Guillou/Segers:2008}.

Specifically, we may start for instance with the following nonparametric
estimator $\widehat{A}$ considered in \cite{Gudendorf/Segers:2011},
\cite{Zhang/Wells/Peng:2008}, \cite{Caperaa/Fougeres/Genest:1997},
\cite{Deheuvels:1991} and \cite{Pickands:1981}. Suppose that the dependence
structure of $\mathbf{X}\in\mathbb{R}^{d}$ ($i=1,\ldots,n$) is characterized
by an extreme value copula $C\left(  \cdot,\mathbf{\theta}\right)
\in\mathcal{C}_{p}$, and, as before, the marginals are denoted by
$F_{1},\ldots,F_{d}$. Given $n$ iid realizations $\mathbf{X}_{i}$
($i=1,2,\ldots,n$), define $\mathbf{Y}_{i}=\left(  Y_{i,1},\ldots
,Y_{i,d}\right)  ^{T}$ by
\[
Y_{i,j}=-\log F_{j}\left(  X_{i,j}\right)  .
\]
Then $Y_{i,j}$ are standard exponential random variables, and
\begin{align*}
P\left(  Y_{i,1}>y_{1},\ldots,Y_{i,d}>y_{d}\right)   &  =C\left(  e^{-y_{1}%
},\ldots,e^{-y_{d}},\mathbf{\theta}\right)  =\exp(-\ell(\mathbf{y}%
,\mathbf{\theta}))\\
&  =\exp(-\left\|  \mathbf{y}\right\|  _{1}A(\mathbf{w},\mathbf{\theta}))
\end{align*}
for $\mathbf{y}\in\lbrack0,\infty)^{d}$ (and $\mathbf{w}=\mathbf{y}/\left\|
\mathbf{y}\right\|  _{1}\in\Delta_{d}$). Hence, for%
\[
\xi_{i}(\mathbf{w}):=\min_{j=1,..,d}\frac{Y_{i,j}}{w_{j}}\quad(i=1,\ldots,n)
\]
with $\mathbf{w}\in\Delta_{d}$ one obtains
\[
P(\xi_{i}(\mathbf{w})>t)=P\left(  Y_{i,1}>w_{1}t,\ldots,Y_{i,d}>w_{d}t\right)
=\exp(-tA(\mathbf{w})).
\]
This means that, for any fixed $\mathbf{w}\in\Delta_{d}$, $\xi_{1}%
(\mathbf{w}),\ldots,\xi_{n}(\mathbf{w})$ are iid exponentially distributed
with mean $1/A(\mathbf{w})$, and $-\log(\xi_{i}(\mathbf{w}))$ ($i=1,\ldots,n$)
are iid Gumbel distributed with location parameter $\log A(\mathbf{w})$. In
particular,%
\[
E(-\log\xi_{i}(\mathbf{w}))=\log A(\mathbf{w})+\gamma
\]
where $\gamma$ is the expectation of the standard Gumbel distribution
(i.e. $\gamma=\Gamma^{\prime}(1)\approx 0.5772$, the Euler-Mascheroni
constant). A nonparametric estimator of $A$ may therefore be defined by (see
\cite{Pickands:1981}, {\small \cite{Deheuvels:1991}, \cite{Hall/Tajvidi:2000},
\cite{Fils-Villetard/Guillou/Segers:2008}, \cite{Gudendorf/Segers:2011})}
\begin{equation}
\log\widehat{A}(\mathbf{w})=-\frac{1}{n}\sum_{i=1}^{n}\log\xi_{i}%
(w)-\gamma\text{ }(\mathbf{w}\in\Delta_{d}). \label{Eq:Initial_Estimator}%
\end{equation}
Unfortunately, $\widehat{A}$ satisfies \eqref{eq:tail_dep_func} 
and~\eqref{Eq:Spectral_Measure_Condition_A} only by chance. 
The standardization~\eqref{Eq:Spectral_Measure_Condition_A} can be achieved by modifications of $\log \widehat{A}$ that substract 
a suitable linear combination of $\log\widehat{A}(\mathbf{w})$ evaluated at certain values of $\mathbf{w}$ (cf.\ \cite{Caperaa/Fougeres/Genest:1997}, \cite{Zhang/Wells/Peng:2008}, \cite{Gudendorf/Segers:2011}).
In particular,
\cite{Gudendorf/Segers:2011} defines a nonparametric least squares
estimator (nonparametric OLS) $\widehat{A}_{OLS}(\mathbf{w})$ by 
$\log\widehat{A}_{OLS}(\mathbf{w})=\hat{\beta}_{0}(\mathbf{w})$ 
where $\hat{\beta}_{0}(\mathbf{w})$ is obtained by least squares regression of
$\log\xi_{i}(\mathbf{w})-\gamma$ ($i=1,2,...,n$) on $-\log\xi_{i}%
(\mathbf{e}_{1})-\gamma,...,-\log\xi_{i}(\mathbf{e}_{d})-\gamma$. 
The resulting estimate satisfies~\eqref{Eq:Spectral_Measure_Condition_A}, 
but it still may fail to satisfy~\eqref{eq:tail_dep_func}. 

In contrast, in the parametric approach introduced above a modification is not
needed, because the projection on $\mathcal{A}_{p}$ automatically
leads to a proper Pickands dependence function. 
However, the examples below demonstrate that 
the initial estimator remains crucial for both 
the asymptotic distribution and the finite sample behaviour of 
the parametric projection. 
Specifically, we will compare the parametric approach based 
on $\widehat{A}$ and $\widehat{A}_{OLS}$ respectively.

%\textbf{(start changes)}
%Why we are not using the estimator from \cite{Gudendorf/Segers:2012}: 
%\begin{itemize}
%\item It extends to unknown margins (which is good), but it is not the problem discussed 
%here. 
%\item Known margins mean less technicalities that are not related to the main problem 
%studied (dimension reduction)
%\item Extension to unknown margins via \cite{Gudendorf/Segers:2012} is straightforward. 
%If the parametric model satisfies the additional assumptions from 
%\cite{Gudendorf/Segers:2012}, the consistency and asymptotic normality results for the 
%resulting parameter estimates in the setting with unknown margins are straightforward.
%\end{itemize}
%\textbf{(end changes)}

Classical results on empirical processes yield a functional central limit
theorem of the following form (see \cite{Gudendorf/Segers:2011} and references
therein). Let $\mathcal{C}(\Delta_{d})$ denote the Banach space of real-valued
continuous functions on $\Delta_{d}$ equipped with the supremum norm.
%%% Remove this comment later
% \textbf{
% \small
% \hrule
% This is the right space; the results by \cite{Gudendorf/Segers:2011} 
% are for $\mathcal{C}(\Delta_d)$. 
% Why not the same symbol $\overset{\mathrm{w}}{\rightarrow}$ for 
% both $\mathcal{C}(\Delta_d)$ and $\mathbb{R}$? We could just emphasize, 
% where it is the case, that convergence is in $\mathcal{C}(\Delta_d)$.
% \hrule}
%%%%%%%%%%%%%%%%%%%%%%%%%%%%%%%
Then, as $n\rightarrow\infty$,
\begin{equation}
\sqrt{n}\left(  \widehat{A}\left(  \mathbf{w}\right)  
-A\left(  \mathbf{w}\right) \right)  
\overset{\mathrm{w}}{\rightarrow}A\left(  \mathbf{w}\right)
\zeta\left(  \mathbf{w}\right)  \text{ }=:\zeta_{A,\text{nonp}}(\mathbf{w})
\text{ in }\mathcal{C}(\Delta_{d}) 
\label{FCLT_for_the_raw_estimator}%
\end{equation}
where $\zeta\left(  \mathbf{w}\right)$ 
($\mathbf{w}\in\Delta_{d}$) is a zero
mean Gaussian process with covariance function
\[
\gamma_{\zeta}\left(  \mathbf{v},\mathbf{w}\right)  
:=\mathrm{cov}(\mathbf{\zeta}(\mathbf{v}),\mathbf{\zeta}(\mathbf{w}))
=\mathrm{cov}\left( -\log\xi(\mathbf{v}),-\log\xi(\mathbf{w})\right)  
\text{ \ }(\mathbf{v},\mathbf{w}\in\Delta_{d}).
\]
Note that the joint distribution of $\xi_{i}(\mathbf{v}),\xi_{i}(\mathbf{w})$
does not depend on $i$, so that dropping the index $i$ here does not lead to 
confusion. This result implies in particular that, for large $n$ and
$\mathbf{w}_{1},\ldots,\mathbf{w}_{N}\in\Delta_{d}$, the joint distribution of
$\widehat{A}(\mathbf{w}_{1}),\ldots,\widehat{A}(\mathbf{w}_{N})$ can be
approximated by an $N-$dimensional normal distribution with mean
\[
\mu\left( \mathbf{w}_{1,}\ldots,\mathbf{w}_{N}\right)  
=(A(\mathbf{w}_{1}),\ldots,A(\mathbf{w}_{N}))^{T}
\]
and covariance matrix 
$n^{-1}\Sigma
=n^{-1}\left[ 
\sigma\left(  \mathbf{w}_{i},\mathbf{w}_{j}\right)  
\right]_{i,j=1,\ldots,N}$ 
with
\begin{align}
\sigma\left( \mathbf{w}_{i},\mathbf{w}_{j}\right)   
&  = \nonumber
A(\mathbf{w}_{i})A(\mathbf{w}_{j})
\gamma_{\zeta}\left(  \mathbf{w}_{i},\mathbf{w}_{j}\right) \\
&  = \label{eq:covar_structure}
A(\mathbf{w}_{i})A(\mathbf{w}_{j})
\mathrm{cov}\left(  -\log\xi(\mathbf{w}_{i}),-\log\xi(\mathbf{w}_{j})\right).
\end{align}
Now, 
given a parametric class of spectral measures $\mathcal{P}_{p}$ based on
$\Psi_{1}\ldots,\Psi_{p}$, an estimator of the corresponding Pickands
dependence function 
$A\left(  \cdot,\mathbf{\theta}\right)  \in\mathcal{A}_{p}$ 
can be defined as follows. Let $\widehat{A}$ be the preliminary
estimator in (\ref{Eq:Initial_Estimator}), and denote by $A_{i}$
($i=1,\ldots,p$) the Pickands dependence functions corresponding to the 
spectral measures $\Psi_{i}$ ($i=1,\ldots,p$). Define a grid of 
$\mathbf{w}$-values $\mathbf{w}_{1}^{0},\ldots,\mathbf{w}_{N}^{0}\in\Delta_{d}$, 
and the vector 
$\mathbf{\hat{a}}^{0}=\left( \hat{a}_{1}^{0},\ldots,\hat{a}_{N}^{0}\right)^{T}$ 
with $\hat{a}_{i}^{0}=\hat{a}(\mathbf{w}_{i}^{0})$ ($i=1,\ldots,N$) and
\[
\hat{a}(\mathbf{w})=\widehat{A}\left(  \mathbf{w}\right)  -A_{p}\left(
\mathbf{w}\right)  ,\quad\mathbf{w}\in\Delta_{d}.
\]
Furthermore, define the $N\times\left(  p-1\right)$ matrix 
$H^{0}=\left[ h_{j}(\mathbf{w}_{i}^{0})\right]_{i=1,...,N;j=1,...,p-1}$ with
\[
h_{j}(\mathbf{w})
=A_{j}\left( \mathbf{w}\right)-A_{p}\left(  \mathbf{w}\right),
\quad\mathbf{w}\in\Delta_{d}.
\]
Then the least squares estimator of $\mathbf{\theta}$ is equal to
\[
\mathbf{\hat{\theta}}=Q^{0}\mathbf{\hat{a}}^{0}
\]
where $Q^{0}=(  {H^{0}}^{T}{H^{0}})^{-1}{H^{0}}^{T}$. Since the
dependence functions $A_{j}\left(  \mathbf{w}\right)$ are defined for all
$\mathbf{w}\in\Delta_{d}$, an estimate of $A\left( \mathbf{w}\right)$ is
available for any $\mathbf{w}\in\Delta_{p}$ by setting
\[
A(  \mathbf{w},\mathbf{\hat{\theta}})  
=A_{p}\left( \mathbf{w}\right)  
+\mathbf{h}^{T}\left(  \mathbf{w}\right)  
\mathbf{\hat{\theta}}
=A_{p}\left(  \mathbf{w}\right)  
+\mathbf{h}^{T}\left( \mathbf{w}\right) Q^{0}\mathbf{\hat{a}}^{0}
\]
where 
$\mathbf{h}\left( \mathbf{w}\right)  
=\left(  h_{1}(\mathbf{w}),\ldots,h_{p-1}(\mathbf{w})\right)^{T}$. 

Letting $N$ tend to infinity, an
estimator of $\mathbf{\theta}$ based on all values in $\Delta_{d}$ can be
obtained as follows. Suppose that the grid 
$\mathbf{w}_{1}^{0},\ldots,\mathbf{w}_{N}^{0}$ is chosen such that, as 
$N\rightarrow\infty$,
the point measure 
$M_{N}(B)
=
N^{-1}\sum_{i=1}^{N} 1\left\{ \mathbf{w}_{i}^{0} \in B \right\}$ 
($B\in\mathcal{B} (\Delta_{d})$) converges weakly 
to a probability measure $M$ on $\Delta_{d}$ with Lebesgue density
$m(\cdot)$. 
This can be achieved by deterministic or by random choice (by sampling
from $M$) of the grid. 
%(In the latter case, convergence in the sup-norm is achieved almost surely.) 
Then
\[
N^{-1}\left(  {H^{0}}^{T}{H^{0}}\right)  _{i,j}
=
N^{-1}\sum_{l=1}^{N}%
h_{i}(\mathbf{w}_{l}^{0})h_{j}(\mathbf{w}_{l}^{0})
=
\int_{\Delta_d} h_{i}(\mathbf{w}) h_{j}(\mathbf{w}) d M_N(\mathbf{w})
\]
converges to
\[
s_{i,j}
=\int_{\Delta_{d}}h_{i}(\mathbf{w})h_{j}(\mathbf{w}) m(\mathbf{w})d\mathbf{w}.
\]
This follows from the continuity of all $A_j$ (and hence all $h_j$)  
and from the compactness of $\Delta_d$.
Similarly, the limit of
\[
N^{-1}\left(  {H^{0}}^{T}\mathbf{\hat{a}}^{0}\right)  _{j}=N^{-1}\sum
_{l=1}^{N}h_{j}(\mathbf{w}_{l}^{0})\hat{a}(\mathbf{w}_{l}^{0})
\]
is
\[
r_{j}=\int_{\Delta_{d}}h_{j}(\mathbf{w})\hat{a}(\mathbf{w}) m(\mathbf{w})d\mathbf{w}.
\]
Thus we obtain an estimator that depends on the density function $m$,
\[
\mathbf{\hat{\theta}}_{M}=S^{-1}\mathbf{r}
\]
where $S=\left(  s_{i,j}\right)  _{i,j=1,\ldots,p-1}$ and 
$\mathbf{r}=\left( r_{1},\ldots,r_{p-1}\right)^{T}$. 
Even more generally, the previous estimators
can be seen as special cases of
\begin{equation}
\mathbf{\hat{\theta}}_{M}=S^{-1}\mathbf{r} 
\label{Eq:New_Estimator_1}
\end{equation}
where
\begin{align}
s_{i,j}  &  
=\int_{\Delta_{d}}
h_{i}(\mathbf{w})h_{j}(\mathbf{w})dM(\mathbf{w}), \nonumber\\
r_{j}  &  
=\int_{\Delta_{d}}h_{j}(\mathbf{w})\hat{a}(\mathbf{w})dM(\mathbf{w})
\nonumber
\end{align}
and $M$ is any distribution function on $\Delta_{d}$ such that $S$ is of full rank.

The same approach can be applied to any initial nonparametric estimator of $A$
for which a functional limit theorem is available. In particular, for
the nonparametric OLS, $\widehat{A}_{OLS}$, Gudendorf and Segers
(\cite{Gudendorf/Segers:2011}) obtain
\begin{align*}
%\label{FCLT}
\sqrt{n}\left( \widehat{A}_{OLS}\left( \mathbf{w}\right)  
-A\left( \mathbf{w}\right)  \right)  
&\overset{\mathrm{w}}{\rightarrow}
A\left( \mathbf{w}\right)  
\left[  \zeta\left(  \mathbf{w}\right)  
-\lambda_{opt}^{T}(\mathbf{w})\mathbf{\zeta}\left( \mathbf{e}\right) \right] \\
&=:\zeta_{A,OLS,\text{nonp}}(\mathbf{w})
\text{ in }\mathcal{C}(\Delta_{d})
\end{align*}
where $\zeta\left( \mathbf{w}\right)$ is the Gaussian process defined in
(\ref{FCLT_for_the_raw_estimator}), 
$\mathbf{\zeta}\left(  \mathbf{e}\right)
=\left(  \zeta\left(  \mathbf{e}_{1}\right),...,
\zeta\left(  \mathbf{e}_{d}\right)  \right)^{T}$, 
$\lambda_{opt}(\mathbf{w})
=\Sigma^{-1}E\left[ \mathbf{\zeta}\left(  \mathbf{e}\right) \zeta(\mathbf{w})\right]$ 
and
$\Sigma=$ 
$E\left[ \mathbf{\zeta}\left( \mathbf{e}\right)  
\mathbf{\zeta}^{T}\left( \mathbf{e}\right) \right]$. 
Applying the parametric approach, we define as before
\begin{equation}
\mathbf{\hat{\theta}}_{M,OLS}=S^{-1}\mathbf{r}_{OLS}
\label{Eq:Second_New_Estimator_1}
\end{equation}
with
\begin{equation*}
r_{OLS,j}=\int_{\Delta_{d}}h_{j}(\mathbf{w})\hat{a}_{OLS}(\mathbf{w})dM(\mathbf{w}) 
\end{equation*}
and
\[
\hat{a}_{OLS}(\mathbf{w})
=\widehat{A}_{OLS}\left(  \mathbf{w}\right)
-A_{p}\left(  \mathbf{w}\right),\quad\mathbf{w}\in\Delta_{d}.
\]

\section{Asymptotic results}

We will use the notation 
$\sigma\left(  \mathbf{v},\mathbf{w}\right)
=\mathrm{cov}\left(  
\zeta_{A,\text{nonp}}(\mathbf{v}),\zeta_{A,\text{nonp}}
(\mathbf{w})\right)$ 
and 
$\sigma_{OLS}\left(  \mathbf{v},\mathbf{w}\right)
=\mathrm{cov}\left(  \zeta_{A,OLS,\text{nonp}}(\mathbf{v}),
\zeta_{A,OLS,\text{nonp}}(\mathbf{w})\right)$ 
for the asymptotic covariance functions of $\hat{A}$
and $\hat{A}_{OLS}$ respectively. The asymptotic distributions of
$\mathbf{\hat{\theta}}_{M}$ and $\mathbf{\hat{\theta}}_{M,OLS}$ are given by

\begin{theorem}
\label{thm:1} 
Let $\mathbf{X}_{i}=\left( X_{i,1},\ldots,X_{i,d}\right)^{T}\in\mathbb{R}^{d}$ 
($i=1,2,\ldots,n$) be iid realizations of a
$d$-dimensional random vector $X$ with marginal distributions 
$F_{1},\ldots,F_{d}$ and extreme value copula 
$C\left( \cdot,\mathbf{\theta}^{0}\right)\in\mathcal{C}_{p}$. 
Denote by $A_{j}$ ($j=1,\ldots,p$) the Pickands dependence functions 
defining $\mathcal{C}_{p}$,
and let $\mathbf{\hat{\theta}}_{M}$ be defined by (\ref{Eq:New_Estimator_1})
and $\mathbf{\hat{\theta}}_{M,OLS}$ by (\ref{Eq:Second_New_Estimator_1}),
where $M$ is such that $S$ is of full rank. Suppose furthermore that
%
%\textbf{
the parameter $\mathbf{\theta}$ is identifiable and 
%that }
%
$\mathbf{\theta}^{0}$ is in the interior of the parameter space 
$\Theta
=\left\{  \mathbf{\theta}=\left(  \theta_{1},\ldots,\theta_{p-1}\right)^{T}
\in\mathbb{R}_{+}^{p-1}:\left\|  \mathbf{\theta}\right\|_{1}\leq1\right\}$. 
Then, as $n\rightarrow\infty$, $\mathbf{\hat{\theta}}_{M}$
and $\mathbf{\hat{\theta}}_{M,OLS}$ converge to $\mathbf{\theta}^{0}$ in
probability, and
\begin{align*}
&  \sqrt{n}\left( \mathbf{\hat{\theta}}_{M}-\mathbf{\theta}^{0}\right)
\overset{\mathrm{w}}{\rightarrow}\mathbf{Z,}\\
&  \sqrt{n}\left( \mathbf{\hat{\theta}}_{M,OLS}-\mathbf{\theta}^{0}\right)
\overset{\mathrm{w}}{\rightarrow}\mathbf{Z}_{OLS}
\end{align*}
where $\mathbf{Z}$ and $\mathbf{Z}_{OLS}$ are $(p-1)$-dimensional normal
random vectors with zero mean and covariance matrices
\begin{align*}
\mathrm{cov}(\mathbf{Z})  &  =V=S^{-1}\Omega \left(S^{-1}\right)^{T},\\
\mathrm{cov}(\mathbf{Z}_{OLS})  &  =V_{OLS}=S^{-1}\Omega_{OLS}\left(S^{-1}\right)^{T}
\end{align*}
where $\Omega=\left[  \omega_{j,l}\right]  _{j,l=1,\ldots,p-1}$ and
$\Omega_{OLS}=\left[  \omega_{j,l}^{OLS}\right]  _{j,l=1,\ldots,p-1}$ are
defined by
\begin{align}
\omega_{j,l}  
&  :=  \label{eq:covar_integr}
\int_{\Delta_{d}}\int_{\Delta_{d}}h_{j}(\mathbf{v})h_{l}(\mathbf{w})
\sigma(\mathbf{v},\mathbf{w})\,dM(\mathbf{v})dM(\mathbf{w}),\\
\omega_{j,l}^{OLS}  
& :=  \nonumber
%\label{eq:covar_OLS_integr}
\int_{\Delta_{d}}\int_{\Delta_{d}}h_{j}%
(\mathbf{v})h_{l}(\mathbf{w})\sigma_{OLS}(\mathbf{v},\mathbf{w})\,dM(\mathbf{v})dM(\mathbf{w}).
\end{align}
\end{theorem}

\begin{proof} %[Proof of Theorem~\ref{thm:1}] 
Since the proof for $\hat{\theta}_M$ and $\hat{\theta}_{M,OLS}$ 
is the same, it is stated for the first estimator only. 
We have 
\begin{align*}
\hat{\mathbf{\theta}}_M 
&= 
S^{-1} \int_{\Delta_d} \mathbf{h}(\mathbf{w})
\left(
\widehat{A}(\mathbf{w}) - A_p(\mathbf{w})
\right) 
d M(\mathbf{w})\\
&= S^{-1} \phi \left( \widehat{A} - A_p \right),
\end{align*}
where $\phi(f)= \int_{\Delta_d} \mathbf{h}(\mathbf{w}) f(\mathbf{w}) dM(\mathbf{w})$ is a linear mapping from $\mathcal{C}(\Delta_d)$ into $\mathbb{R}^{p-1}$. 
Analogously, we have $\mathbf{\theta}^0 = S^{-1}\phi(A - A_p)$, and hence 
\[
\sqrt{n} \left( \hat{\mathbf{\theta}} - \mathbf{\theta}^0 \right) 
= 
S^{-1} \phi \left( \widehat{A} - A \right). 
\]
It is obvious that the mapping $f\mapsto S^{-1}\phi(f)$ is continuous. 
Hence the functional Central Limit Theorem 
\eqref{FCLT_for_the_raw_estimator} for $\widehat{A}$ 
and the Continuous Mapping Theorem 
yield 
\[
\sqrt{n} \left( \hat{\mathbf{\theta}} - \mathbf{\theta}^{0} \right) \stackrel{\mathrm{w}}{\to} 
S^{-1}\phi \left(\zeta_{A,nonp}\right)
=: \mathbf{Z}.
\]
Recall that $\zeta_{A,nonp}$ is a zero-mean Gaussian process 
with covariance function $\sigma(\mathbf{v},\mathbf{w}) = E [\zeta_{A,nonp} (\mathbf{v}) \zeta_{A,nonp} (\mathbf{w})]$ introduced in~\eqref{eq:covar_structure}. 
Hence, as a linear mapping of $\zeta_{A,nonp}$, 
the random vector $\mathbf{Z}$ is Gaussian with zero mean and  
covariance matrix $V=S^{-1} \Omega (S^{-1})^{T}$,
where $\Omega=[\omega_{i,j}]_{i,j=1,\ldots,1-p}$ 
is the covariance matrix of $\phi(\zeta_{A,nonp})$. 
The representation~\eqref{eq:covar_integr} follows from Fubini's Theorem: 
\begin{align*}
\omega_{i,j} 
&= 
E \left[
\int_{\Delta_d}h_i (\mathbf{v}) \zeta_{A,nonp} (\mathbf{v}) dM(\mathbf{v})
\int_{\Delta_d}h_j(\mathbf{w}) \zeta_{A,nonp} (\mathbf{w}) dM(\mathbf{w})
\right] \\
&=
\int_{\Delta_d}\int_{\Delta_d} h_i(\mathbf{v}) h_j (\mathbf{w}) 
E [\zeta_{A,nonp} (\mathbf{v}) \zeta_{A,nonp} (\mathbf{w})] 
dM(\mathbf{v}) dM(\mathbf{w}).
\end{align*} 
%which is precisely~\eqref{eq:covar_integr}.
\end{proof} 

An immediate consequence of this result is the asymptotic normality of
$A(  \mathbf{w},\mathbf{\hat{\theta}}_{M})  
=A_{p}(\mathbf{w})+\mathbf{h}^{T}\left(  \mathbf{w}\right)  \mathbf{\hat
{\theta}}_{M}$ and 
$A(  \mathbf{w},\mathbf{\hat{\theta}}_{M,OLS})
=A_{p}\left(  \mathbf{w}\right)  +\mathbf{h}^{T}\left(  \mathbf{w}\right)
\mathbf{\hat{\theta}}_{M,OLS}$
uniformly in $\mathbf{w}\in\Delta_d$.

\begin{corollary}
\label{cor:1} Under the assumptions of Theorem~\ref{thm:1} we have, as
$n\rightarrow\infty$,
\begin{equation}
\sqrt{n}\left(  A\left(  \mathbf{w},\mathbf{\hat{\theta}}_{M}\right)
-A\left(  \mathbf{w},\mathbf{\theta}^{0}\right)  \right)  \overset{\mathrm{w}%
}{\rightarrow}\zeta_{A}\left(  \mathbf{w}\right)  , \label{eq:AN_for_A(.)}%
\end{equation}%
\begin{equation*}
\sqrt{n}\left(  A\left(  \mathbf{w},\mathbf{\hat{\theta}}_{M,OLS}\right)
-A\left(  \mathbf{w},\mathbf{\theta}^{0}\right)  \right)  \overset{\mathrm{w}%
}{\rightarrow}\zeta_{A,OLS}\left(  \mathbf{w}\right)
\end{equation*}
where $\zeta_{A}\left(  \mathbf{w}\right)  $, $\zeta_{A,OLS}\left(
\mathbf{w}\right)  $ ($\mathbf{w\in}\Delta_{d}$) are zero-mean Gaussian
process with covariance functions%
\[
\gamma_{A}\left(  \mathbf{v},\mathbf{w}\right)  =\mathbf{h}^{T}\left(
\mathbf{v}\right)  V\mathbf{h}\left(  \mathbf{w}\right)  ,
\]
and
\[
\gamma_{A}\left(  \mathbf{v},\mathbf{w}\right)  =\mathbf{h}^{T}\left(
\mathbf{v}\right)  V_{OLS}\mathbf{h}\left(  \mathbf{w}\right)  ,
\]
where $V$ and $V_{OLS}$ are as in Theorem~\ref{thm:1}. 
\end{corollary}

\begin{proof}%[Proof of Corollary~\ref{cor:1}]
Recall that 
\begin{equation*}
A( \mathbf{w},\mathbf{\hat{\theta}}_{M})
=A_{p}( \mathbf{w}) 
+ \mathbf{h}^{T}(\mathbf{w})\mathbf{\hat{\theta}}_{M}.
\end{equation*}
As the mapping $\mathbf{\theta}\mapsto 
\mathbf{h}^T\mathbf{\theta}$ 
is linear and continuous in $\mathcal{C}(\Delta_d)$, 
we obtain~\eqref{eq:AN_for_A(.)} 
from the Continuous Mapping Theorem. 
In fact, we have the representation 
$\zeta_A(\mathbf{w})=\mathbf{h}^T(\mathbf{w})\mathbf{Z}$ 
in~\eqref{eq:AN_for_A(.)}. 
The covariance structure of the limit process follows from 
\begin{equation*}
\mathrm{cov}
\left( 
\mathbf{h}^{T}( \mathbf{v}) 
\mathbf{\hat{\theta}}_{M},
\mathbf{h}^{T}( \mathbf{w}) 
\mathbf{\hat{\theta}}_{M}\right) 
=
\mathbf{h}^{T}%
( \mathbf{v}) 
\mathrm{var} (\mathbf{\hat{\theta}}_{M}) 
\mathbf{h}( \mathbf{w}).
\end{equation*}
\end{proof}%

Note that, more specifically, Theorem \ref{thm:1} implies that 
$\sqrt{n}( A(  \mathbf{w},\mathbf{\hat{\theta}}_{M})  
-A(\mathbf{w},\mathbf{\theta}^{0}) )$ 
and 
$\sqrt{n}(
A(\mathbf{w},\mathbf{\hat{\theta}}_{M,OLS})  
-A(\mathbf{w},\mathbf{\theta}^{0}))$ 
are asymptotically equivalent to the stochastic processes 
$\zeta_{A}(\mathbf{w})=\mathbf{h}^{T}(\mathbf{w})\mathbf{Z}$ 
and 
$\zeta_{A,OLS}(\mathbf{w})=\mathbf{h}^{T}(\mathbf{w})\mathbf{Z}_{OLS}$, respectively, 
with index $\mathbf{w}\in\Delta_d$. The random variables $\mathbf{Z}$ and 
$\mathbf{Z}_{OLS}$ are the weak limits in Theorem \ref{thm:1}.

Another consequence of Theorem~\ref{thm:1} is the asymptotic distribution 
of $C(\mathbf{u},\mathbf{\hat{\theta}}_{M})$ and 
$C(\mathbf{u},\mathbf{\hat{\theta}}_{M,OLS})$.

\begin{corollary}
\label{cor:2}Under the assumptions of Theorem~\ref{thm:1} we have, as
$n\rightarrow\infty$,%
\begin{align*}
&  \sqrt{n}\left(  C(\mathbf{u},\mathbf{\hat{\theta}}_{M})
-C(\mathbf{u},\mathbf{\theta}^{0}) \right)  
\overset{\mathrm{w}}{\rightarrow}\zeta_{C},\\
&  \sqrt{n}\left(  C(\mathbf{u},\mathbf{\hat{\theta}}_{M,OLS})
-C(\mathbf{u},\mathbf{\theta}^{0})  \right)  
\overset{\mathrm{w}}{\rightarrow}\zeta_{C,OLS}
\end{align*}
where $\zeta_{C}$, $\zeta_{C,OLS}$ are zero mean Gaussian processes with
covariance functions%
\begin{align*}
\gamma_{C}(  \mathbf{u},\mathbf{v})   &  
=\dot{C}^{T}( \mathbf{u},\mathbf{\theta}^{0}) 
V\dot{C}(\mathbf{v},\mathbf{\theta}^{0}),\\
\gamma_{C,OLS}(  \mathbf{u},\mathbf{v})   &  
=\dot{C}^{T}(\mathbf{u},\mathbf{\theta}^{0})  
V_{OLS}\dot{C}(\mathbf{v},\mathbf{\theta}^{0})  .
\end{align*}
Here,
\[
\dot{C}(\cdot,\mathbf{\theta}^{0})  
=C(  \cdot,\mathbf{\theta}^{0})  
\left( \log\frac{C_{1}(\cdot)}{C_{p}(  \cdot)},
\ldots,\log\frac{C_{p-1}(\cdot)}{C_{p}(\cdot)}\right)^{T}
\]
and $V$, $V_{OLS}$ are as in Theorem~\ref{thm:1}. More specifically,
\[
\zeta_{C}(\cdot)=
\dot{C}^{T}(\cdot,\mathbf{\theta}^{0})\mathbf{Z,}
\text{ }\zeta_{C,OLS}(\cdot)=\dot{C}^{T}(\cdot,\mathbf{\theta}^{0})\mathbf{Z}_{OLS}%
\]
with $\mathbf{Z}$, $\mathbf{Z}_{OLS}$ from Theorem~\ref{thm:1}.
\end{corollary}

\begin{proof}%[Proof of Corollary~\ref{cor:2}]
Recall~\eqref{eq:C(.,theta)} and denote
$r=r(\mathbf{u}):=\|-\log \mathbf{u}\|_1$ and  $\mathbf{w}=\mathbf{w}(\mathbf{u})
:=r^{-1}(-\log \mathbf{u})$, 
where $\log(\mathbf{u})$ is understood componentwise.  
Then we obtain  
\[
C(\mathbf{u},{\theta})
=
\exp\left(-r \left( 
\mathbf{\theta}^T \mathbf{h}(\mathbf{w}) + A_p(\mathbf{w}) 
\right) \right),
\]
and hence, for $j,k=1,\ldots,p-1$,
\begin{eqnarray*}
\partial_{\theta_j} C(\mathbf{u},\mathbf{\theta}) 
&=&
- r h_j(\mathbf{w})
\exp\left(-r \left( 
\mathbf{\theta}^T \mathbf{h}(\mathbf{w}) + A_p(\mathbf{w}) 
\right)\right)\\
\partial_{\theta_j\theta_k}
C(\mathbf{u},\mathbf{\theta}) 
&=&
r^2 h_j(\mathbf{w}) h_k(\mathbf{w}) 
\exp\left(-r \left( 
\mathbf{\theta}^T \mathbf{h}(\mathbf{w}) + A_p(\mathbf{w}) 
\right)\right).
\end{eqnarray*}

A well known consequence of~\eqref{eq:tail_dep_func}
is that each Pickands 
dependence function $A$ assumes values in $[1/d,1]$ only. 
Since $\mathbf{\theta}^T\mathbf{h}(\mathbf{w}) + A_p(\mathbf{w})
= A(\mathbf{w},\mathbf{\theta})$ 
is a proper Pickands dependence function, and each $h_j$ is a difference of 
two, we obtain 
\begin{eqnarray*}
| \partial_{\theta_j} C(\mathbf{u},\mathbf{\theta}) |
&\le&
2r(\mathbf{u}) \exp(-r(\mathbf{u})/d)\\
| \partial_{\theta_j\theta_k} C(\mathbf{u},\mathbf{\theta}) | 
&\le&
4r^2(\mathbf{u}) \exp(-r(\mathbf{u})/d).
\end{eqnarray*}

Thus the first- and second-order derivatives of 
$C(\mathbf{u},\mathbf{\theta})$ with respect to $\mathbf{\theta}$ 
are uniformly bounded in $\mathbf{u}\in[0,1]^d$, 
and the Taylor approximation
\[
C(\mathbf{u},\mathbf{\hat{\theta}}) 
-C(\mathbf{u},\mathbf{\theta}^0)
=
%\sum_{j=1}^{p-1} \partial_{\theta_j}%
%
%
C(\mathbf{u},\mathbf{\theta}^0)(\hat{\theta}_j- \theta^0_j)
\left( \dot{C}\left(\mathbf{u},\mathbf{\theta}^0\right)\right)^T
(\mathbf{\hat{\theta}} - \mathbf{\theta}^0)
+
O( \left\| \mathbf{\hat{\theta}}- \mathbf{\theta}^0 \right\|^2 )
\]
with $\dot{C}(\mathbf{u},\mathbf{\theta})
:=(\partial_{\theta_1} C(\mathbf{u},\mathbf{\theta}),
\ldots,\partial_{\theta_{p-1}}C(\mathbf{u},\mathbf{\theta}))^T$ 
is uniform in $\mathbf{u}\in[0,1]$. 
The final result now easily follows from Theorem~\ref{thm:1}.   
\end{proof}

\section{Examples and simulations}

\subsection{A parametric model example} \label{sec:simul_example}

To illustrate how one may construct spectral measures 
$\Psi_{1},\ldots,\Psi_{p}$ 
for building parametric models, we consider 
%the following 
an
example with $d=2$ and $p=3$. 
For ease of notation, we parametrize $\Delta_2$ by the first coordinate, 
so that $t\in[0,1]$ represents $(t,1-t)\in\Delta_2$, and we can write 
$d\Psi(t)$, $A(t)$, etc. 
In particular, if a spectral measure $\Psi$ has a Lebesgue density $f$, then 
the standardization~\eqref{Eq:Spectral_Measure_Condition} reads as 
\begin{equation} \label{Eq:Spectral_Measure_Condition_f}
1= \int_{0}^{1} t f(t) dt = \int_{0}^{1} (1-t) f(t) dt. 
\end{equation}
Let
\begin{align*}
f_{1}(t)  &  =%
\begin{cases}
\frac{a}{2}(1-\cos(3\pi t)) & x\in\lbrack0,2/3]\\
ab(1+\cos(3\pi t/2)) & t\in(2/3,1]
\end{cases}
\\
f_{2}(t)  &  =f_{1}(1-t)\\
f_{3}(t) & =c\sin(\pi t).
\end{align*}
With appropriate constants $a$, $b$, and $c$, the functions $f_i$ satisfy~\eqref{Eq:Spectral_Measure_Condition_f}, and we define the basis elements of 
the parametric family by $d\Psi_i(t) = f_i(t) dt$.  
Figure~\ref{Fig:example_f_i} shows plots of the spectral densities $f_i$, $f_2$, $f_3$. 
The constants $a,b,c$ are derived as follows.
\begin{figure}[htbp]
\centering
\includegraphics[width=0.8\textwidth]{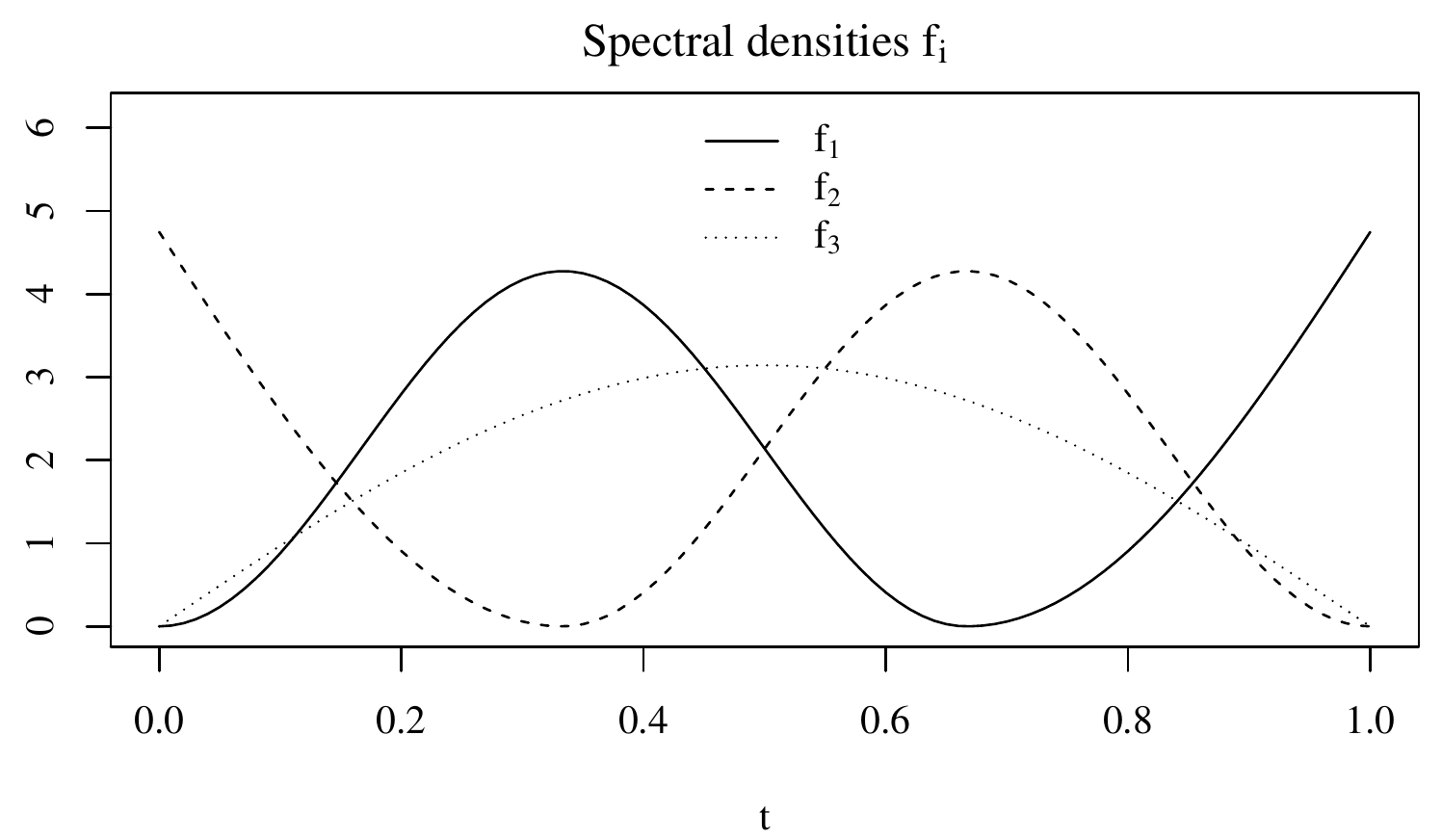}
\caption{Spectral densities $f_1$, $f_2$, and $f_3$.}
\label{Fig:example_f_i}
\end{figure}

Due to symmetry, we have $\int_{0}^{1} tf_{3}(t)dt=\int_{0}^{1}(1-t)f_{3}(t)dt$, 
so that~\eqref{Eq:Spectral_Measure_Condition_f} yields 
\[
2=\int_{0}^{1}f_{3}(t)dt=c\int_{0}^{1}\sin\left(  \pi t\right)  dt=\frac{2c}%
{\pi}.
\]
Thus, $c=\pi$.

To determine $a$ and $b$ in $f_{1}$ and $f_{2}$, it suffices to consider $f_1$. 
Let $g_{1}(z):=\int_{0}^{z} tf_{1}(t)dt$. For $z\in\lbrack0,2/3]$ one has%
\begin{align*}
g_{1}(z)  &  =\frac{a}{2}\left(  \left[  \frac{t^{2}}{2}\right]  _{0}%
^{z}-\left[  \frac{t\sin(3\pi t)}{3\pi}\right]  _{0}^{z}-\left[
\frac{\cos[3\pi t]}{9\pi^{2}}\right]  _{0}^{z}\right) \\
&  =\frac{a}{2}\left(  \frac{z^{2}}{2}-\frac{z\sin(3\pi z)}{3\pi}%
-\frac{\cos(3\pi z)}{9\pi^{2}}+\frac{1}{9\pi^{2}}\right)  ,
\end{align*}
and for $z\in(2/3,1]$,%
\begin{align*}
g_{1}(z)  &  =\frac{a}{9}+ab\left(  \left[  \frac{t^{2}}{2}\right]  _{2/3}%
^{z}+\left[  \frac{t\sin(3\pi t/2)}{3\pi/2}\right]  _{2/3}^{z}+\left[
\frac{\cos(3\pi t/2)}{9\pi^{2}/4}\right]  _{2/3}^{z}\right) \\
&  =\frac{a}{9}+ab\left(  \frac{z^{2}}{2}-\frac{2}{9}+\frac{z\sin(3\pi
z/2)}{3\pi/2}+\frac{\cos(3\pi z/2)}{9\pi^{2}/4}+\frac{1}{9\pi^{2}/4}\right)  .
\end{align*}
In particular,%
\[
g_{1}(1)=\frac{a}{9}+ab\left(  \frac{5}{18}-\frac{2}{3\pi}+\frac{4}{9\pi^{2}%
}\right)  .
\]
Moreover, note that 
$\int_{0}^{z}(1-t)f_{1}(t)dt=h_{1}(z)-g_{1}(z)$ 
with $h_{1}(z):=\int_{0}^{z}f_{1}(t)dt$. 
For $z\in\lbrack0,2/3]$ one obtains
\[
h_{1}(z)=\frac{a}{2}\left(  z-\frac{\sin(3\pi z)}{3\pi}\right)  ,
\]
and for $z\in(2/3,1]$,
\begin{align*}
h_{1}(z)  &  =\frac{a}{3}+ab\left[  t+\frac{\sin(3\pi t/2)}{3\pi/2}\right]
_{2/3}^{z}\\
&  =\frac{a}{3}+ab\left(  z-2/3+\frac{\sin(3\pi z/2)}{3\pi/2}\right)  .
\end{align*}
In particular,
\[
h_{1}(1)=\frac{a(1+b(1-2/\pi))}{3}.
\]
We need $a$ and $b$ such that $f_1$ satisfies~\eqref{Eq:Spectral_Measure_Condition_f}, which is equivalent to $g_{1}(1)=1$ and $h_1(1)=2$. 
The latter equation yields $a=6\left(  1+b(1-2/\pi)\right)
^{-1}$. Substituting this in $g_{1}(1)=1$,
\[
b=\frac{\pi^{2}}{8-6\pi+2\pi^{2}}%
\]
and hence%
\[
a=\frac{12\pi^{2}-36\pi+48}{3\pi^{2}-8\pi+8}.
\]
Note that further spectral densities of this type can be defined, for instance, 
by replacing $2/3$ in the definition of $f_{1}$ by other values (in the
interval $(0,1)$).

\subsection{Sampling technique}

The asymptotic results obtained in section 4 are illustrated by
simulations for the example introduced above. Thus, the copula 
$C(\cdot,\theta)$ is defined by the Pickands dependence function
\[
A(t,\mathbf{\theta})=\theta_{1}A_{1}(t)+\theta_{2}%
A_{2}(t)+(1-\theta_{1}-\theta_{2})A_{3}(t)
\]
with $t\in[0,1]$ representing $(t,1-t)\in\Delta_2$, 
$0<\theta_{1},\theta_{2}<1$, $\theta_{1}+\theta_{2} \le 1$, 
and $A_i$ $i=1,2,3$ being the Pickands dependence functions 
corresponding to $\Psi_{i}$ (and $f_i$).   
A random vector $(X_{1},X_{2})\sim C(\cdot,\theta)$ 
can be simulated exactly using the algorithm proposed in
\cite{Ghoudi_et_al:1998}. More specifically, given a bivariate 
dependence function $A=A(\cdot,\theta)$, the corresponding extreme
value copula $C=C(\cdot,\theta)$ can be sampled as follows:

\begin{enumerate} 
\item Simulate $Z\in\lbrack0,1]$ with distribution function
\[
P(Z\le z) = z+z(1-z)\frac{A^{\prime}(z)}{A(z)} =: G_Z(z).
\]
Note that if $Z$ has a density $g_{Z}$, then $g_{Z}=G_{Z}^{\prime}$. 
\item Calculate
\[
p(Z)=\frac{Z(1-Z)A^{\prime\prime}(Z)}{A(Z)G_Z^{\prime}(Z)}.
\]
Let $V=U_{1}$ with probability $p(Z)$ and $V=U_{1}U_{2}$
with probability $1-p(Z)$, where $U_{1},U_{2}$ are independent and 
uniformly distributed on $[0,1]$.

\item Set $X_{1}=V^{Z/A(Z)}$ and $X_{2}=V^{(1-Z)/A(Z)}$. Then the distribution 
function of the random vector $(X_{1},X_{2})$ is equal to $C$.
\end{enumerate}

The computation of $G_{Z}$, $G_{Z}^{\prime}$, and $p(Z)$ can be simplified as
follows. Recall that, given a spectral density 
$f=\theta_{1}f_{1}+\theta_{2}f_{2}+(1-\theta_{1}-\theta_{2})f_{3}$, 
$A(z)=\int_{0}^{1}\max(tz,(1-t)(1-z))f(t)dt$. 
Since $tz>(1-t)(1-z)$ is equivalent to $t>1-z$, we
obtain
\begin{align}
A(z)  
&  = \nonumber
\int_{0}^{1-z}(1-z)(1-t)f(t)dt+\int_{1-z}^{1}ztf(t)dt\\
&  = \nonumber
(1-z)(h(1-z)-g(1-z))+z(1-g(1-z))\\
&  = \label{eq:A_via_g_and_h}
z-g(1-z)+(1-z)h(1-z),
\end{align}
where $g(z)=\int_{0}^{z}tf(t)dt$ and $h(z)=\int_{0}^{z}f(t)dt$ 
(note that \eqref{Eq:Spectral_Measure_Condition_f} implies $g(1)=1$ and 
$h(1)=2$). 
From~\eqref{eq:A_via_g_and_h} we obtain that 
\begin{equation*}
A^{\prime}(z)=1+g^{\prime}(1-z)-(1-z)h^{\prime}(1-z)-h(1-z).
\end{equation*}
Consequently, $g^{\prime}(z)=zf(z)$ and $h^{\prime}(z)=f(z)$ imply that
\begin{align*}
A^{\prime}(z)  &  =1+(1-z)f(1-z)-(1-z)f(1-z)-h(1-z)\\
&  =1-h(1-z),\\
A^{\prime\prime}(z)  &  =h^{\prime}(1-z)=f(1-z).
\end{align*}
Thus we obtain%
\[
G_{Z}(z)=\frac{zA(z)-(1-z)zA^{\prime}(z)}{A(z)}=\frac{z(1-g(1-z))}{A(z)},
\]
and therefore
\[
G_{Z}^{\prime}(z)=\frac{1-g(1-z)+z(1-z)f(1-z)}{A(z)}%
-\frac{z(1-g(1-z))(1-h(1-z))}{A^{2}(z)}.
\]

It is obvious that all functions $f,A,g,h$ corresponding to $C(\cdot,\theta)$
are convex combinations of the corresponding $f_{i},A_{i},g_{i},h_{i}$ with
weights $\theta_{1}$,$\theta_{2}$, and $1-\theta_{1}-\theta_{2}$. 
Since $f(z)$ ($z\in\lbrack0,1]$) is a convex combination
of the bounded functions $f_{i}$ ($i=1,2,3$), we can simulate $G_{Z}$ by
rejection sampling. Consequently, for the simulation of $(Y_{1},Y_{2})\sim
C(\cdot,\theta)$ we only need $\theta$ and the functions $f_{i},A_{i}%
,g_{i},h_{i}$ for $i=1,2,3$. The representations of $f_1$, $g_{1}$, 
and $h_{1}$ are already derived above, in subsection~\ref{sec:simul_example}. 
An explicit representation for $A_1$ follows from~\eqref{eq:A_via_g_and_h}.
Due to $f_{2}(z)=f_{1}(1-z)$ one obtains $A_{2}%
(z)=A_{1}(1-z)$. Furthermore,
\begin{align*}
h_{2}(z)  &  =\int_{0}^{z}f_{1}(1-t)dt=\int_{1-z}^{1}f_{1}(y)dy=h_{1}%
(1)-h_{1}(1-z)\\
&  =2-h_{1}(1-z),\\
g_{2}(z)  &  =\int_{1-z}^{1}(1-y)f_{1}(y)dy=h_{1}(1)-h_{1}(1-z)-(g_{1}%
(1)-g_{1}(1-z))\\
&  =1-h_{1}(1-z)+g_{1}(1-z).
\end{align*}
Finally, for $i=3$ we have
\begin{align*}
h_{3}(z)  &  =\int_{0}^{z}f_{3}(x)dx=1-\cos(\pi z)\\
g_{3}(z)  &  =\int_{0}^{z}xf_{3}(x)dx=\frac{1}{\pi}\sin(\pi z)-z\cos(\pi z).
\end{align*}

Figures \ref{fig:example_points}a) through d) show typical samples 
$\mathbf{X}_{i}=(X_{i1},X_{i2})\sim C(\cdot,\theta)$ ($i=1,2,...,n$) with $n=1000$, and
$\mathbf{\theta}=(1,0)$, $(0,1)$, $(0,0)$ and $(\frac{1}{2},\frac{1}{2})$
respectively. Image (and contour) plots of two-dimensional kernel density
estimates for these samples are shown in figures \ref{fig:example}a) through 2d).

\subsection{Simulation results}

To study the finite sample performance of the estimators of $A$ discussed above, 
the following simulation study was carried out. 
For $\mathbf{\theta}=(0.1,0.1)$, $(0.05,0.9)$ 
and $(0.8,0.1)$ respectively, $1000$ simulated
samples of size $n=25\cdot2^{j}$ ($j=0,1,...,8$) were generated. For each
sample, the nonparametric estimates $\widehat{A}$ and $\widehat{A}_{OLS}$ as
well as the corresponding parametric estimates 
$A(\mathbf{w},\mathbf{\hat{\theta}}_{M})$ 
and $A(\mathbf{w},\mathbf{\hat{\theta}}_{M,OLS})$ were calculated.
For $M$ (in $A(\mathbf{w},\mathbf{\hat{\theta}}_{M})$
and $A(\mathbf{w},\mathbf{\hat{\theta}}_{M,OLS})$), 
we used a discrete uniform distribution on the grid 
$\mathbf{w}_{i}=(w_{i1},1-w_{i2})$ ($w_{i1}=0.05\cdot i$, $i=1,2,...,19$). 
As expected, the
naive nonparametric estimator $\hat{A}$ turned out to be clearly inferior to
all other methods. For instance, for $\mathbf{\theta}=(0.1,0.1)$ and $n=50$,
the integrated mean squared error 
$IMSE_{\text{nonp}}=\int_{\Delta_{2}}
E[(\widehat{A}(\mathbf{w})-A(\mathbf{w}))^2]d\mathbf{w}$ 
is $74$ times larger than 
$IMSE_{\text{nonp},OLS}
=\int_{\Delta_{2}}
E[(\hat{A}_{OLS}(\mathbf{w})-A(\mathbf{w}))^2] d\mathbf{w}$, 
and $IMSE_{\text{par}}=\int_{\Delta_{2}}
E[(A(\mathbf{w},\mathbf{\hat{\theta}}_{M}) 
-A(\mathbf{w}))^2] d\mathbf{w}$ 
is almost $7$ times larger than the corresponding quantity
(denoted by $IMSE_{\text{par},OLS}$) for 
$A(\mathbf{w},\mathbf{\hat{\theta}}_{M,OLS})$.  
For larger sample sizes the ratios 
$$r_{\text{nonp}}=IMSE_{\text{nonp}}/IMSE_{\text{nonp},OLS}$$ 
and 
$$r_{\text{par}}=IMSE_{\text{par}}/IMSE_{\text{par},OLS}$$ 
stabilize around the values of
$64$ and $30$ respectively. Moreover, even if we compare the nonparametric
OLS, $\hat{A}_{OLS}$, with the parametric estimator 
$A(\mathbf{w},\mathbf{\hat{\theta}}_{M})$, 
we obtain a ratio of 
$IMSE_{\text{par}}/IMSE_{\text{nonp},OLS}\approx26$ for large sample sizes. 
We may thus conclude that using a good initial nonparametric estimator 
for the parametric method is essential. Detailed results on
$\widehat{A}(\mathbf{w})$ and 
$A(\mathbf{w},\mathbf{\hat{\theta}}_{M})$ 
are therefore omitted, and we focus solely on the comparison
between $\hat{A}_{OLS}$ and 
$A(\mathbf{w},\mathbf{\hat{\theta}}_{M,OLS})$. 
Figures \ref{fig:IMSE_ratio}a), b) and c) show the ratio 
$r=IMSE_{\text{par},OLS}/IMSE_{\text{nonp},OLS}$ 
for the three choices of $\mathbf{\theta}$ as a
function of $n$. In all three cases, $r$ stabilizes around a value below $1$.
The numerical values are given in table 1. As a function of $\mathbf{w}$, 
the relative precision of
$\widehat{A}(\mathbf{w})$ 
compared to $A(\mathbf{w},\mathbf{\hat{\theta}}_{M,OLS})$ 
depends on $\mathbf{w}$ and the shape of $A$. This can be
seen in figures \ref{fig:MSE_ratio}a), b) and c), where simulated values of
\[
r\left(  \mathbf{w}\right)  
=\frac{E\left[  
\left( A( \mathbf{w},\mathbf{\hat{\theta}}_{M,OLS})  
-A(\mathbf{w})\right)^{2}\right]}
{E\left[  \left(  \hat{A}_{OLS}(\mathbf{w})-A(\mathbf{w})\right)^{2}\right] }
\]
are ploted as a function of $w_1$, for different values of $n$. Figures
\ref{fig:A_50_estimates_theta0101}, 
\ref{fig:A_50_estimates_theta00509} and
\ref{fig:A_50_estimates_theta0801},
with estimates of $A$ for $50$ series of length $n=25$ (Fig. a,b) and
$n=200$ (Fig. c,d) respectively, illustrate a further problem with the
nonparametric OLS. For small sample sizes, $\hat{A}_{OLS}$ is often not
exactly convex, which means that it is, with relatively high probability, not a
proper dependence function. By definition, this problem does not occur for
$A(\mathbf{w},\mathbf{\hat{\theta}}_{M,OLS})$. Finally,
boxplots of $\hat{\theta}_{1}$ and $\hat{\theta}_{2}$ for the case with 
$\mathbf{\theta}=(0.05,0.9)$ are given in figures \ref{fig:box}a) and b) 
respectively. One can see in
particular that for small sample sizes the distributions of $\hat{\theta}_{1}$
and $\hat{\theta}_{2}$ are skewed to the right and left respectively. This is
due to $\theta_{1}$ and $\theta_{2}$ being close to the border of the
parameter space, and the restrictions $\theta_{1},\theta_{2}\geq0$ and
$\theta_{1}+\theta_{2} \leq 1$.  

\section{Final remarks}

In this paper we considered estimation of extreme value copulas 
based on parametric models that are defined in terms of the 
spectral measure. This approach is very flexible, and in principle any
type of dependence between extremes can be captured.
The method is not restricted to the
case where the marginal distributions are known, since  
any nonparametric estimator $\hat{A}$
can be used in the projection. Theorem 1 and Corollary 1 apply 
(with $\sigma\left( \mathbf{v},\mathbf{w}\right)$ replaced 
by the corresponding asymptotic covariance function) whenever a functional 
limit theorem of the form given in~\eqref{FCLT_for_the_raw_estimator}
holds for $\hat{A}$. 

An important issue that would need to be addressed in future research 
is the extension to a larger class of copulas. 
In this paper, 
observations were assumed to be generated by an extreme value copula.
In practice, an extreme value copula is usually reached only asymptotically 
(for multivariate maxima). In analogy to nonparametric extreme value copula 
estimators, consistent parametric methods will have to be developed for 
such situations. A further question is model choice, 
i.e. the question how to decide on the number and type 
of spectral measures to be used as a basis. 
For data generated by an extreme value
copula, standard methods such as AIC or BIC 
(\cite{Akaike:1973},\cite{Schwarz:1978}) may be useful. 
In the more general situation where an extreme value copula is only reached in the 
limit, the question is more complex.

\section{Acknowledgements}
This research has been supported in part by the DFG-Research Grant BE 2123/11-1. 
Georg Mainik would like to thank RiskLab, ETH Zurich, for financial support. 

%\section{APPENDIX - Proofs}

%\bibliographystyle{abbrvnat}
%\bibliography{references}

\vfill\eject

\begin{table}[htbp]
	\centering
   \begin{tabular}[c]{|c|c|c|c|}
\hline
$n$ & $\mathbf{\theta}=(0.1,0.1)$ & $\mathbf{\theta}=(0.05,0.9)$ &
$\mathbf{\theta}=(0.8,0.1)$\\
\hline
$25$ &  0.375  &  0.504 & 0.480\\
$50$ & 0.444 &  0.622 & 0.652\\
$100$ & 0.512 &  0.636 & 0.696\\
$200$ & 0.613 &  0.652 & 0.767\\
$400$ & 0.756 &  0.674 & 0.815\\
$800$ & 0.840 &  0.749 & 0.876\\
$1600$ & 0.899 &  0.805 & 0.898\\
$3200$ & 0.903 &  0.868 & 0.900\\
$6400$ & 0.902 &  0.887 & 0.903\\
\hline
\end{tabular}
\caption{$r=IMSE_{{\rm par},OLS}/IMSE_{{\rm nonp},OLS}$ for 
$\mathbf{\theta}=(0.1,0.1)$, $(0.05,0.9)$ and $\mathbf{\theta}=(0.8,0.1)$ 
respectively, and sample sizes $n=25\cdot 2^j$ $(j=0,1,...,8)$.}
\end{table}

\vfill\eject

\begin{figure}[htbp]
	\centering	
\includegraphics[height=0.35\textheight,width=0.4\textwidth]{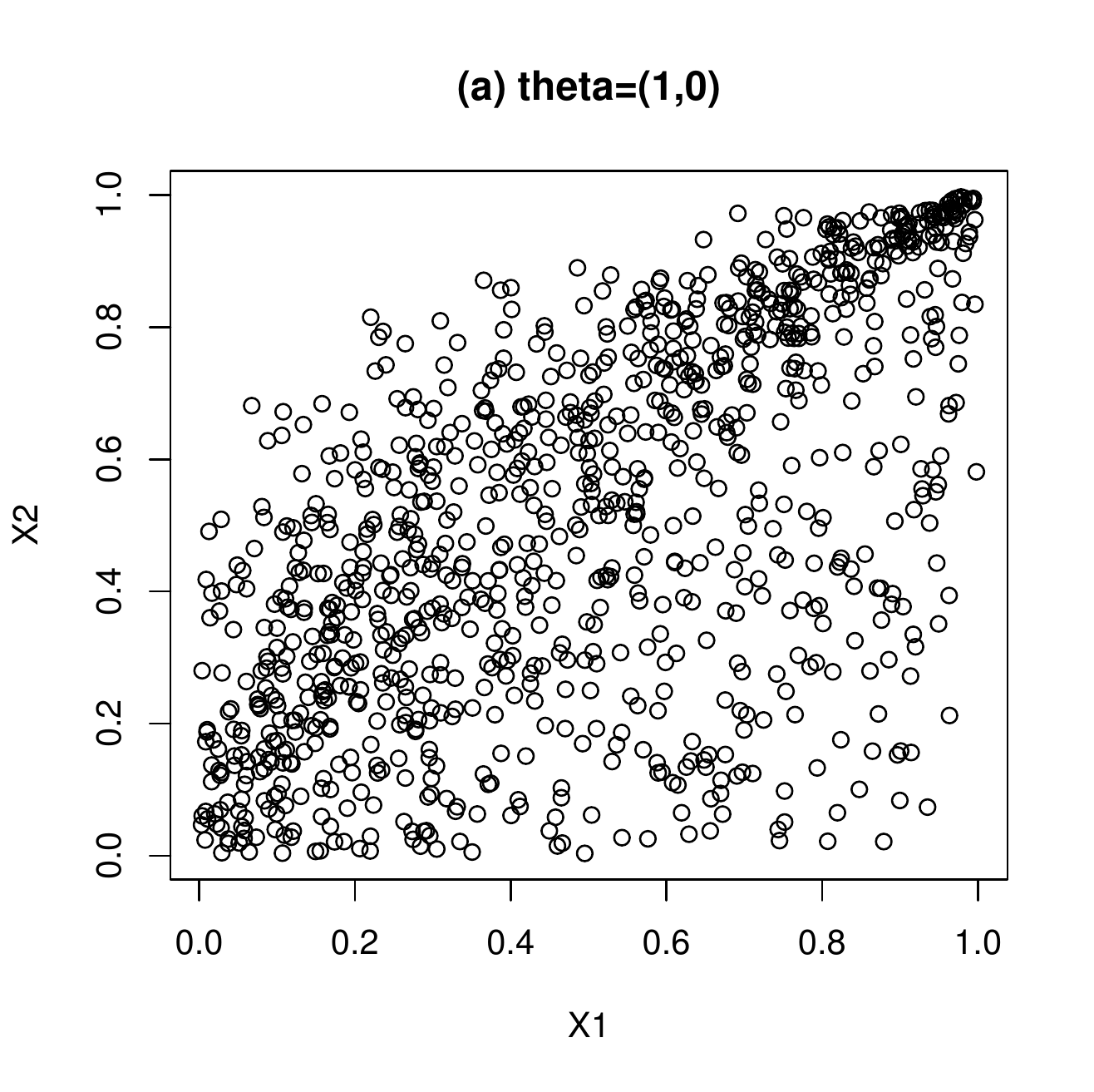}
\includegraphics[height=0.35\textheight,width=0.4\textwidth]{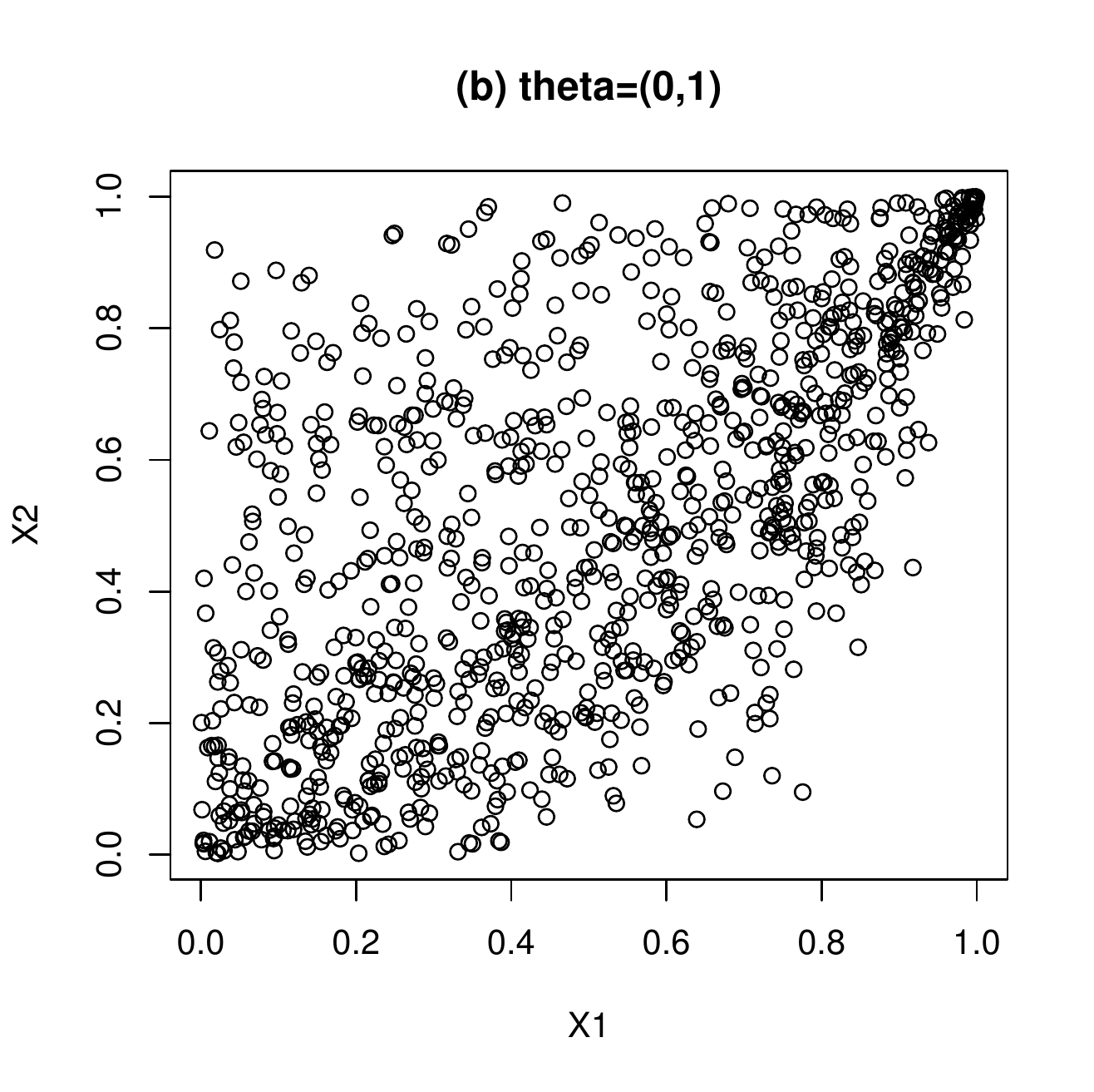}
\includegraphics[height=0.35\textheight,width=0.4\textwidth]{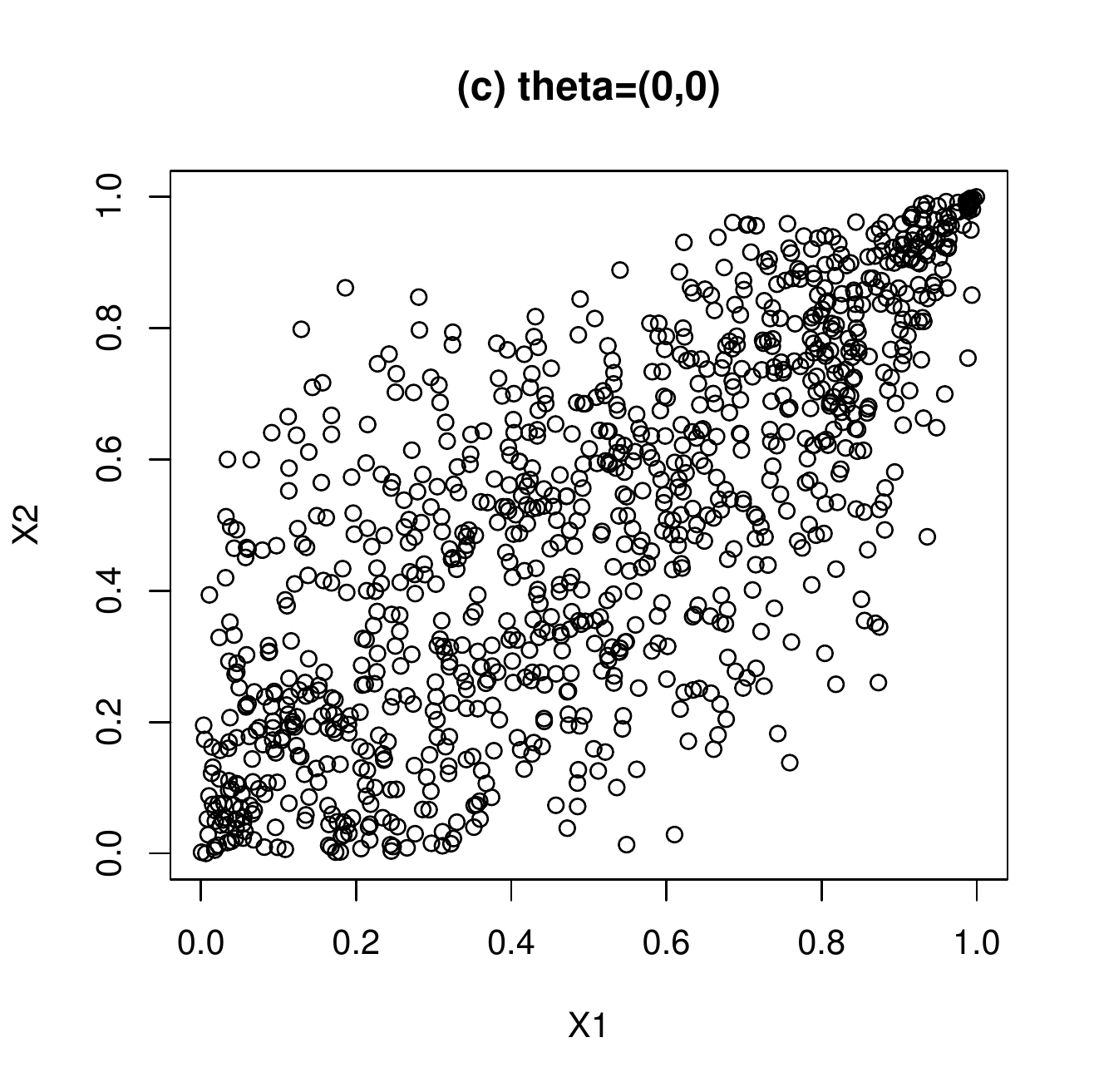}
\includegraphics[height=0.35\textheight,width=0.4\textwidth]{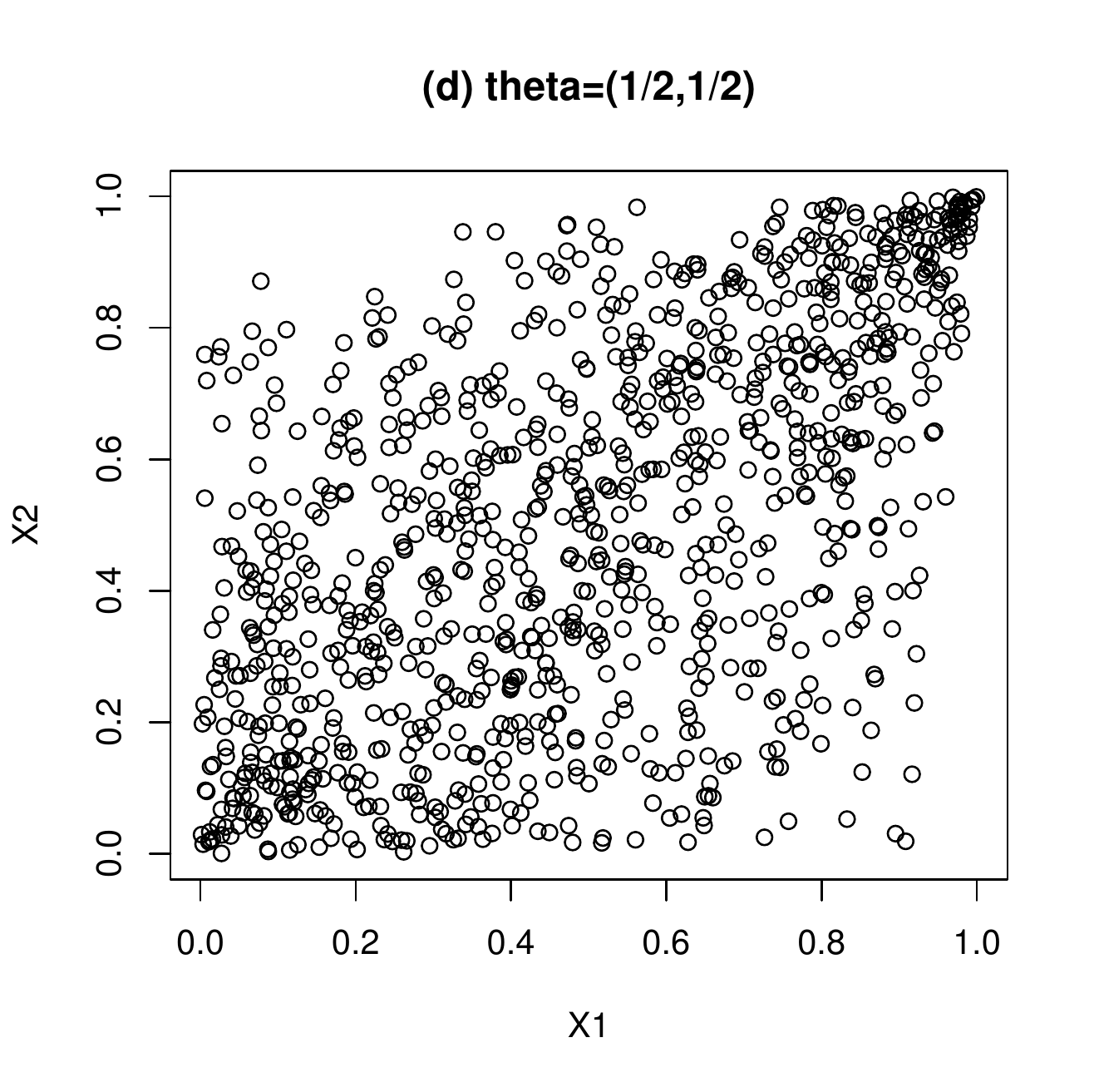}
 	\caption{Simulated samples 
 $\mathbf{X}_{i}=(X_{i1},X_{i2})\sim C(\cdot,\theta)$ ($i=1,2,...,n$) with $n=1000$, and
 $\mathbf{\theta}=(1,0)$, $(0,1)$, $(0,0)$ and $(\frac{1}{2},\frac{1}{2})$
 respectively.}
	\label{fig:example_points}
\end{figure}

\begin{figure}[htbp]
 	\centering	
\includegraphics[height=0.35\textheight,width=0.4\textwidth]{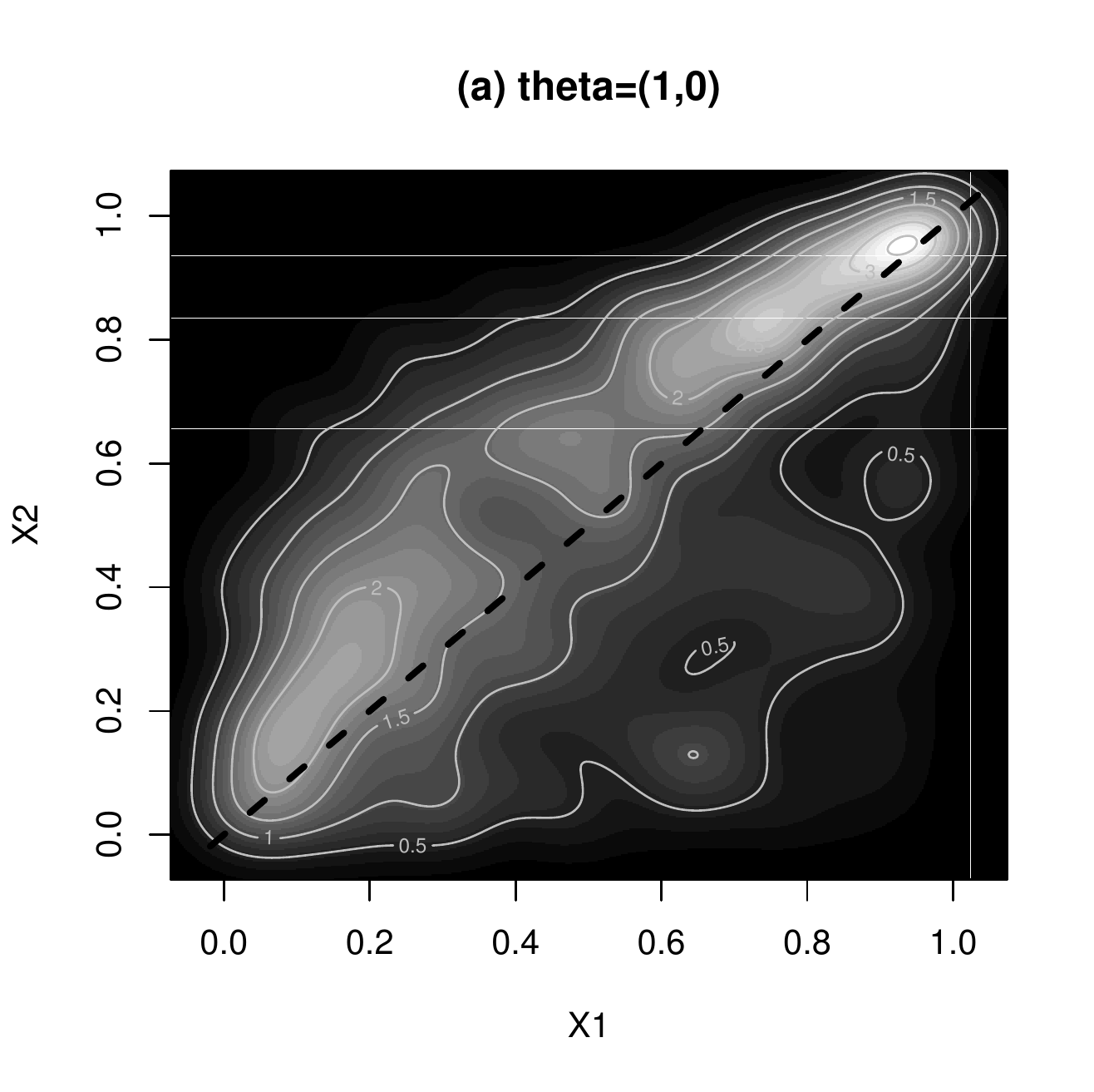}
\includegraphics[height=0.35\textheight,width=0.4\textwidth]{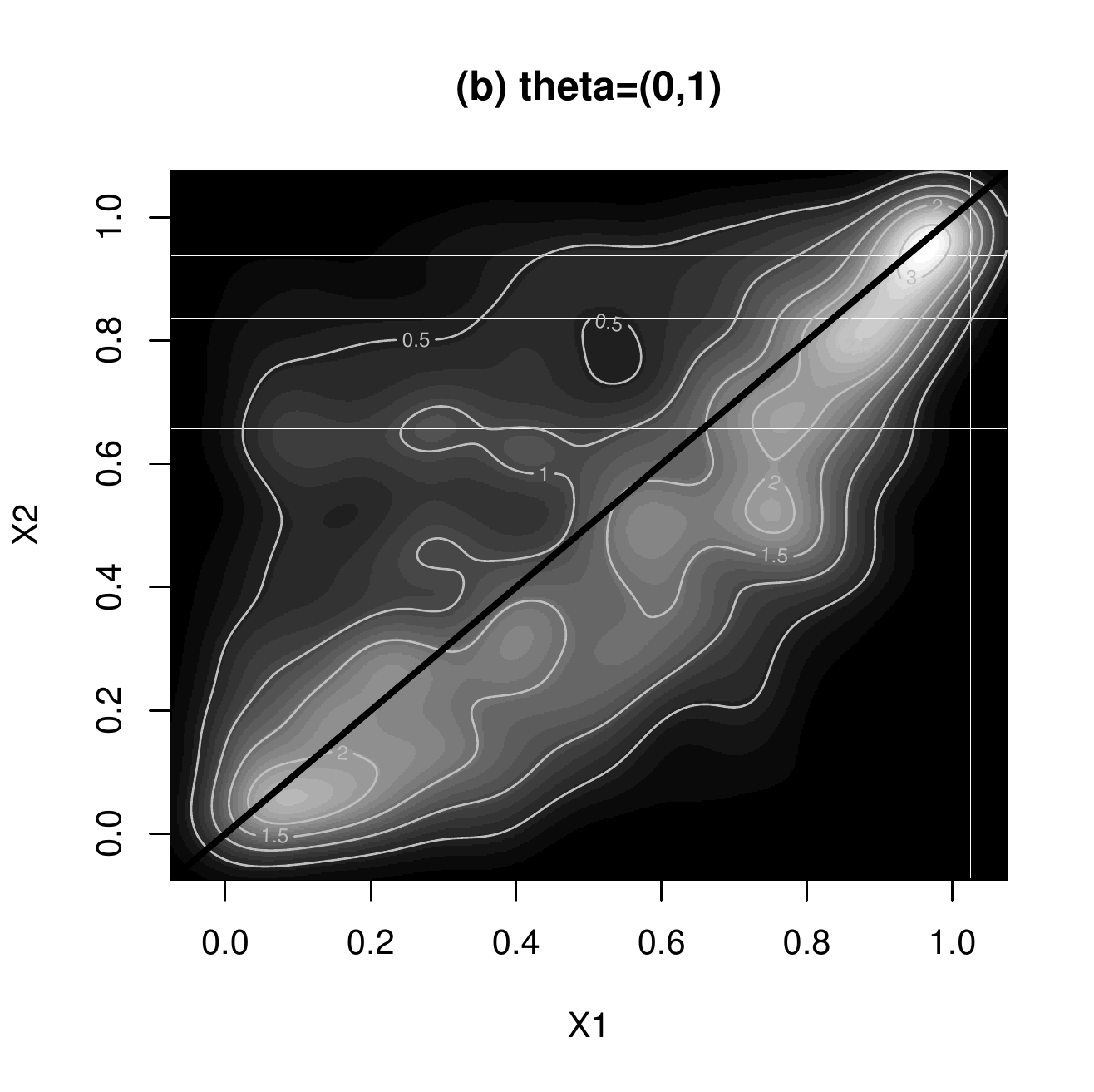}
\includegraphics[height=0.35\textheight,width=0.4\textwidth]{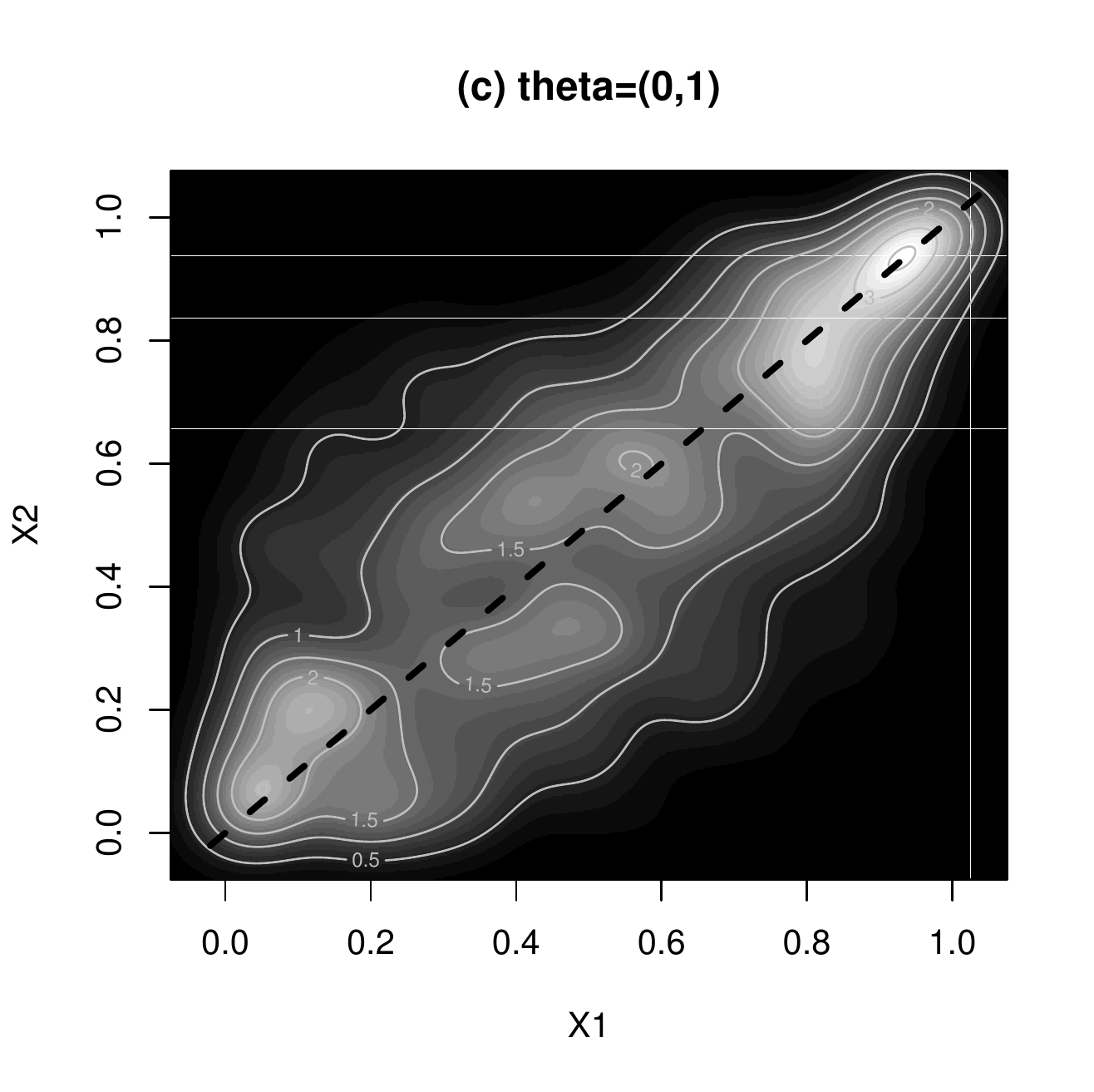}
\includegraphics[height=0.35\textheight,width=0.4\textwidth]{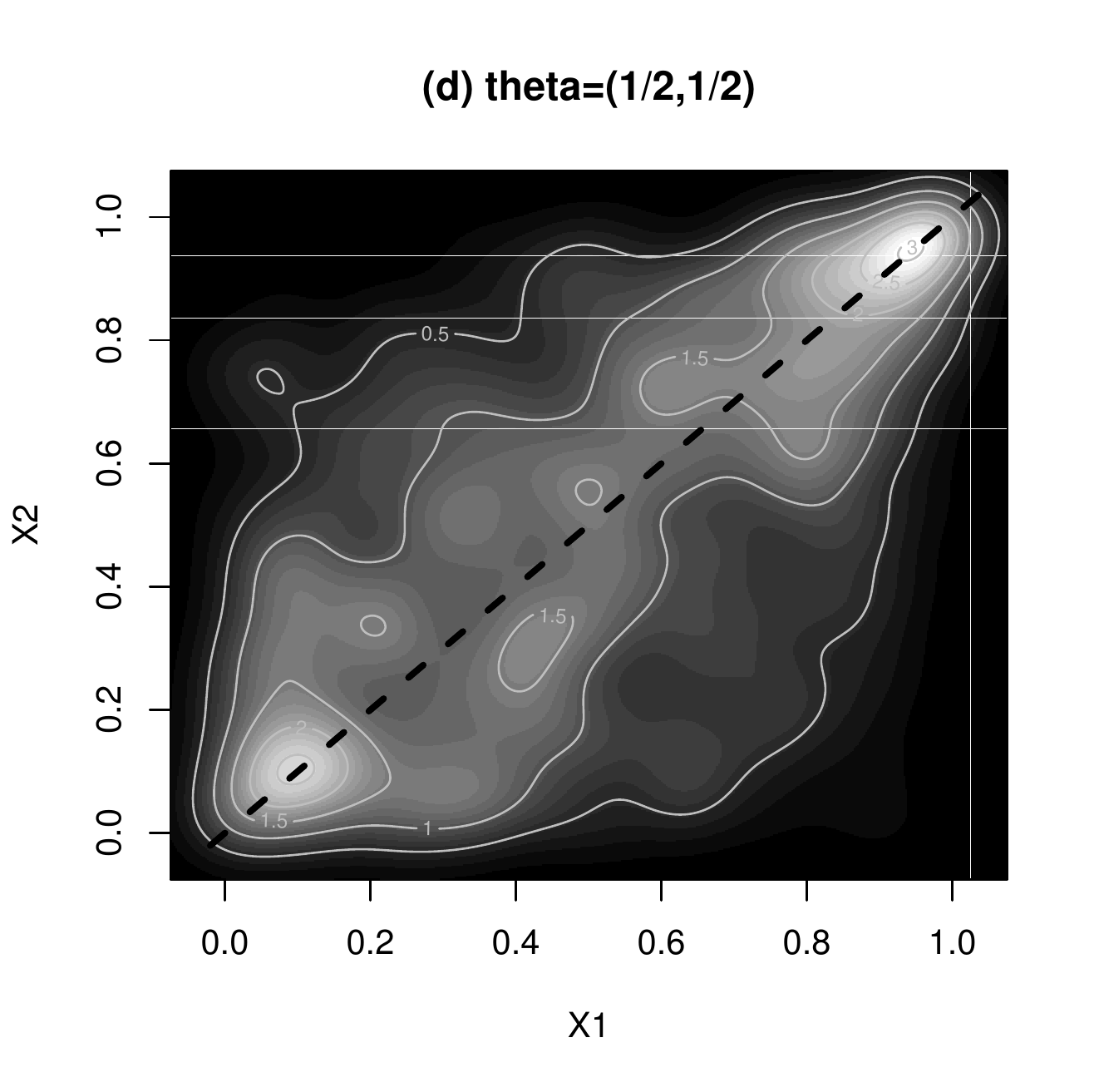}
	\caption{Image and contour plots of nonparametric density estimates 
for the simulated samples in figures \ref{fig:example_points}a) through d).}
	\label{fig:example}
\end{figure}

\begin{figure}[htbp]
	\centering	
\includegraphics[height=0.35\textheight,width=0.4\textwidth]{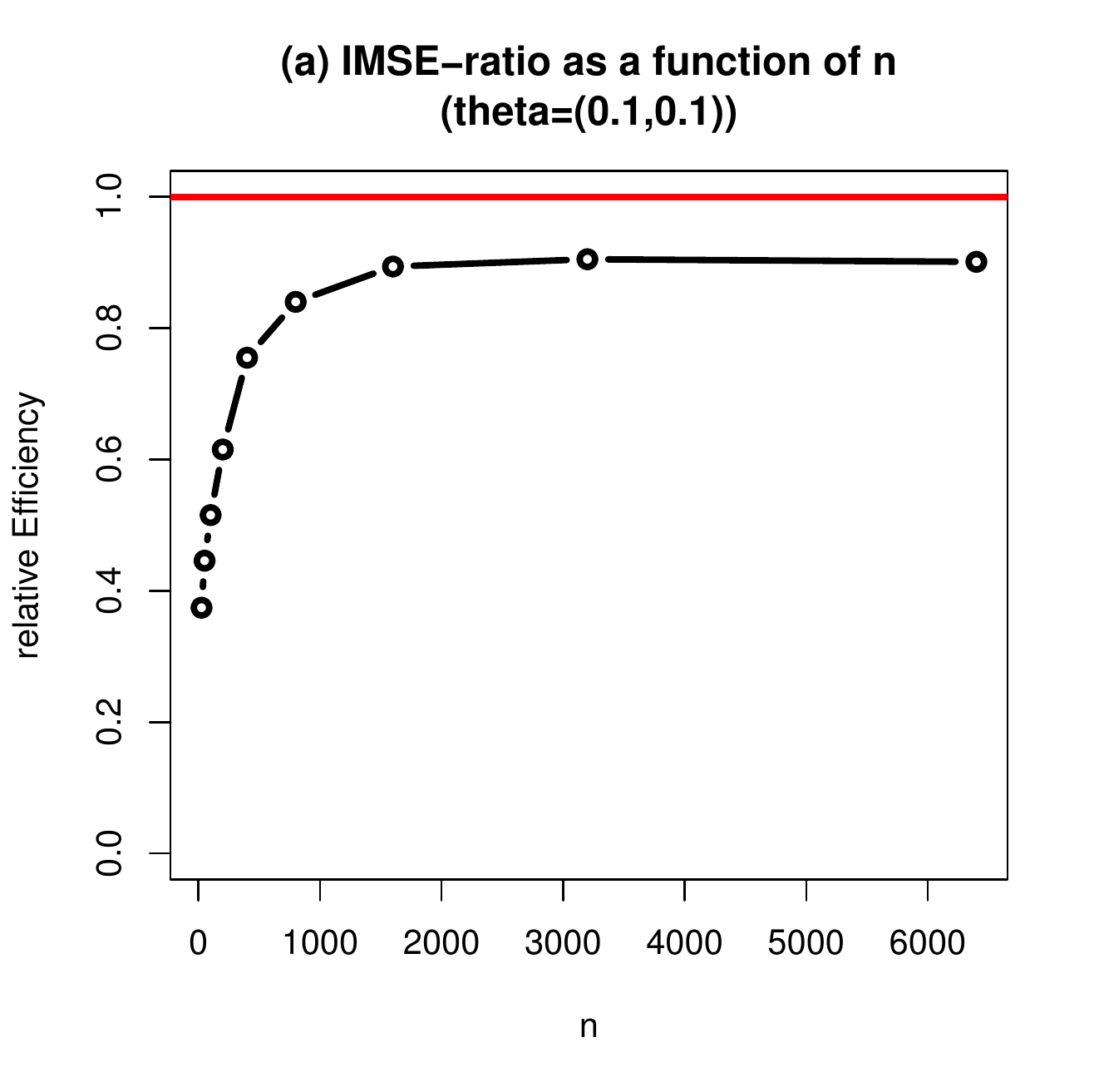}
\includegraphics[height=0.35\textheight,width=0.4\textwidth]{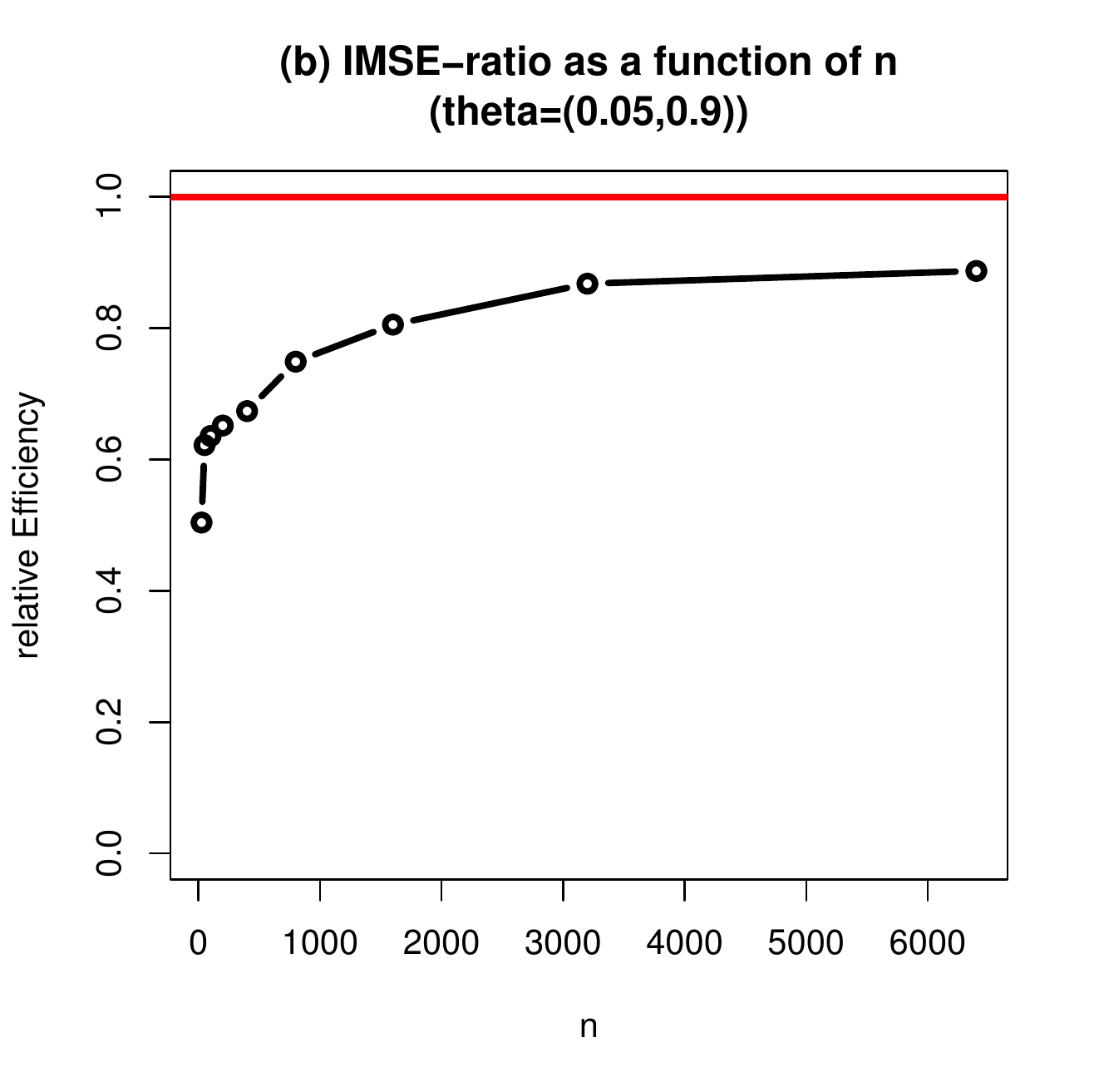}
\includegraphics[height=0.35\textheight,width=0.4\textwidth]{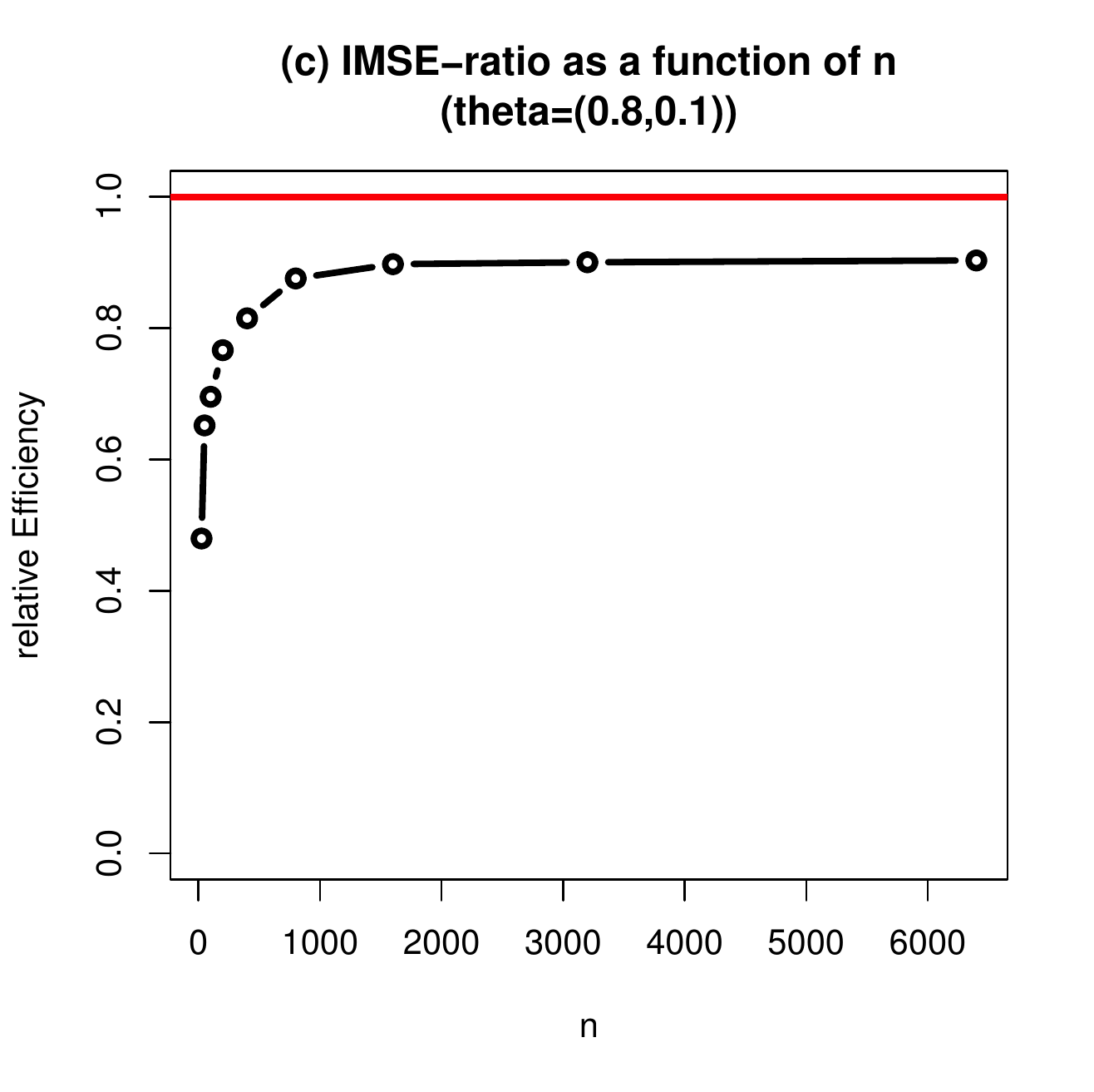}
	\caption{Ratio of parametric and nonparametric IMSE  
	 for (a) $\mathbf{\theta}=(0.1,0.1)$, 
	 (b) $(0.05,0.9)$ and 
	 (c) $(0.8,0.1)$ respectively.}
	\label{fig:IMSE_ratio}
\end{figure}

\begin{figure}[htbp]
	\centering	
\includegraphics[height=0.35\textheight,width=0.4\textwidth]{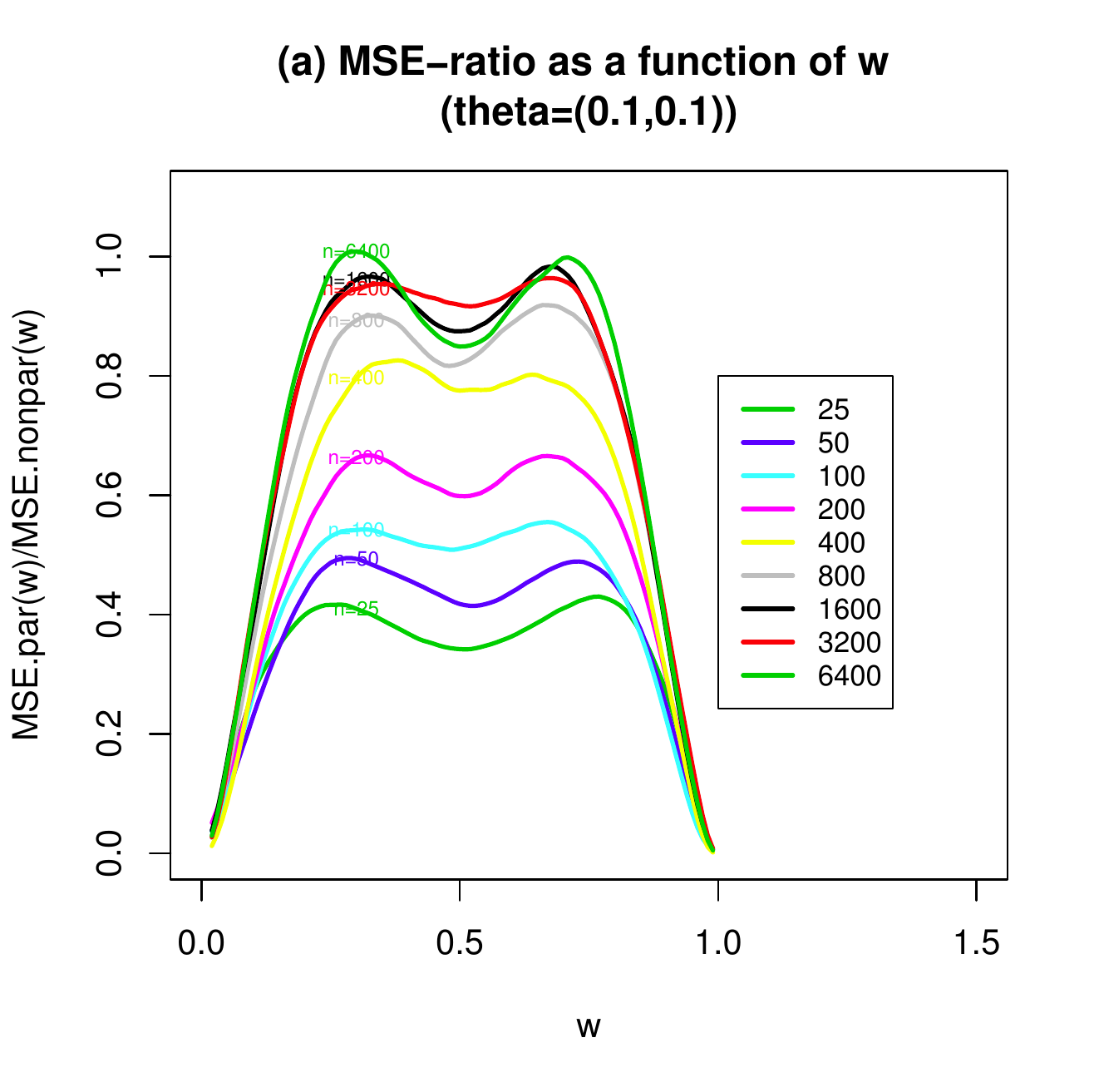}
\includegraphics[height=0.35\textheight,width=0.4\textwidth]{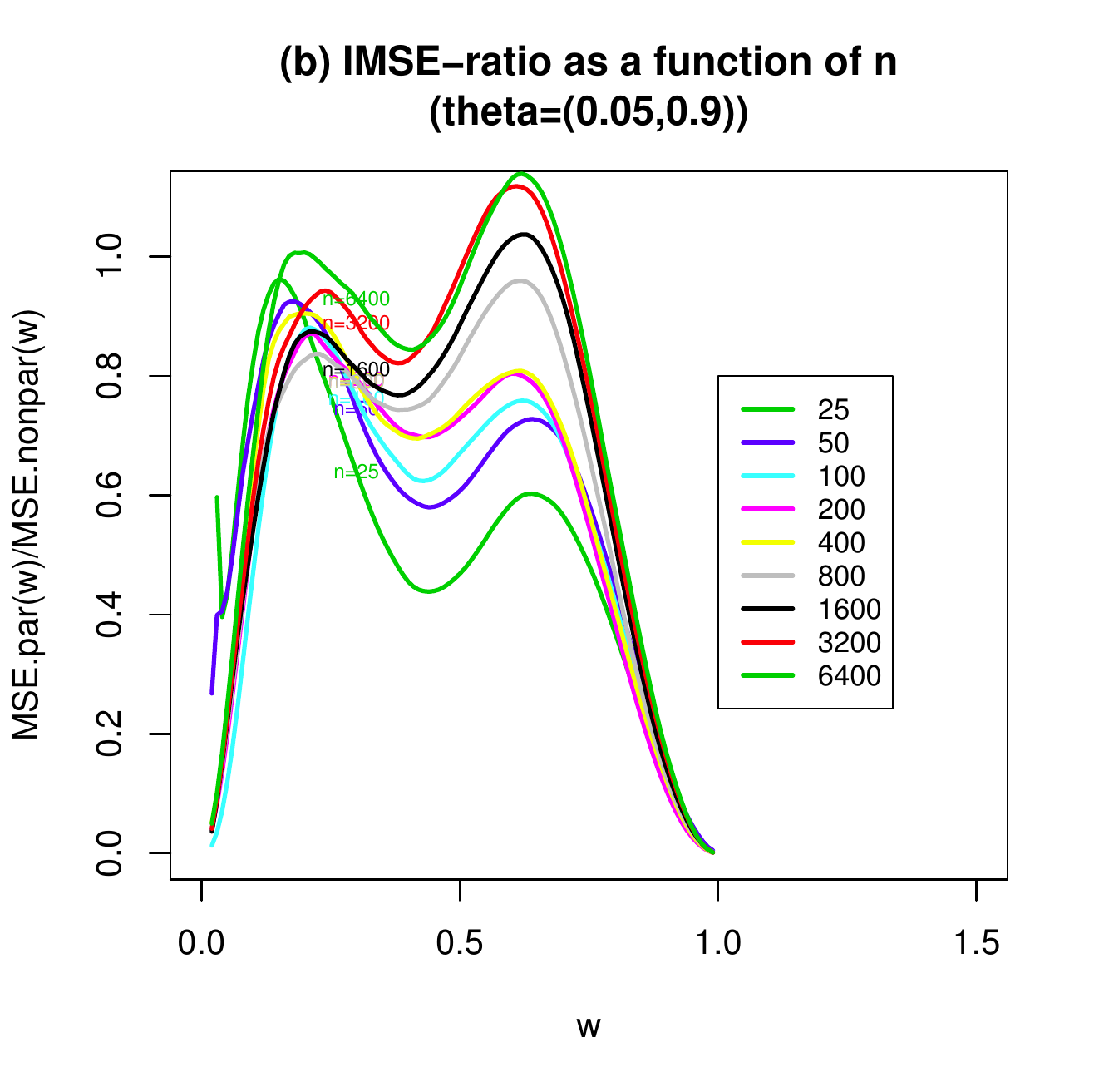}
\includegraphics[height=0.35\textheight,width=0.4\textwidth]{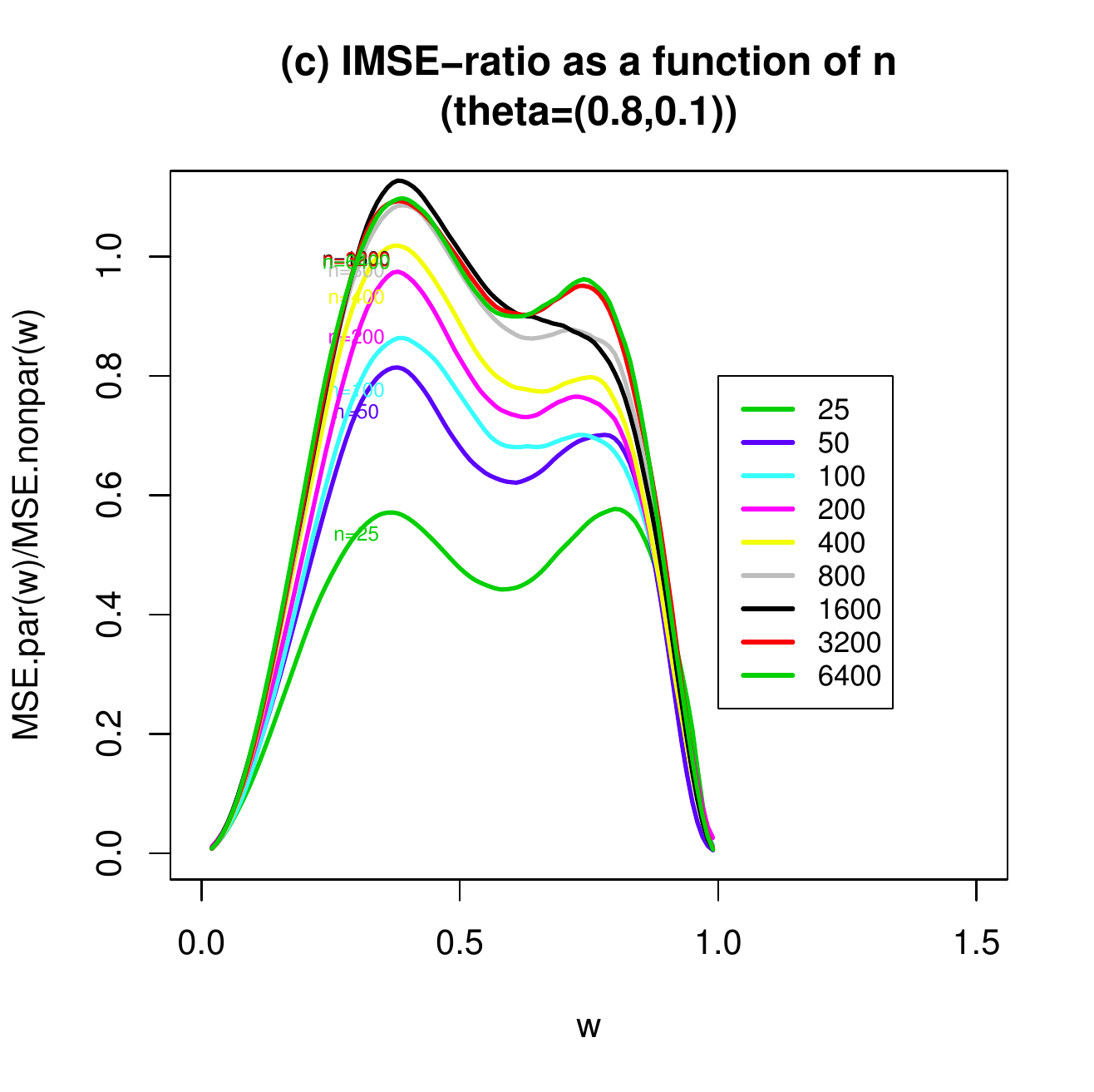}
	\caption{Ratio of parametric and nonparametric MSE, $r\left(  \mathbf{w}\right)$, 
	 for (a) $\mathbf{\theta}=(0.1,0.1)$, 
	 (b) $(0.05,0.9)$ and 
	 (c) $(0.8,0.1)$, and sample sizes $n=$25, 50, 100, 200, 400, 800, 1600, 
	 3200 and 6400 respectively.}
	\label{fig:MSE_ratio}
\end{figure}

\begin{figure}[htbp]
	\centering
\includegraphics[height=0.5\textheight,width=0.9\textwidth]
{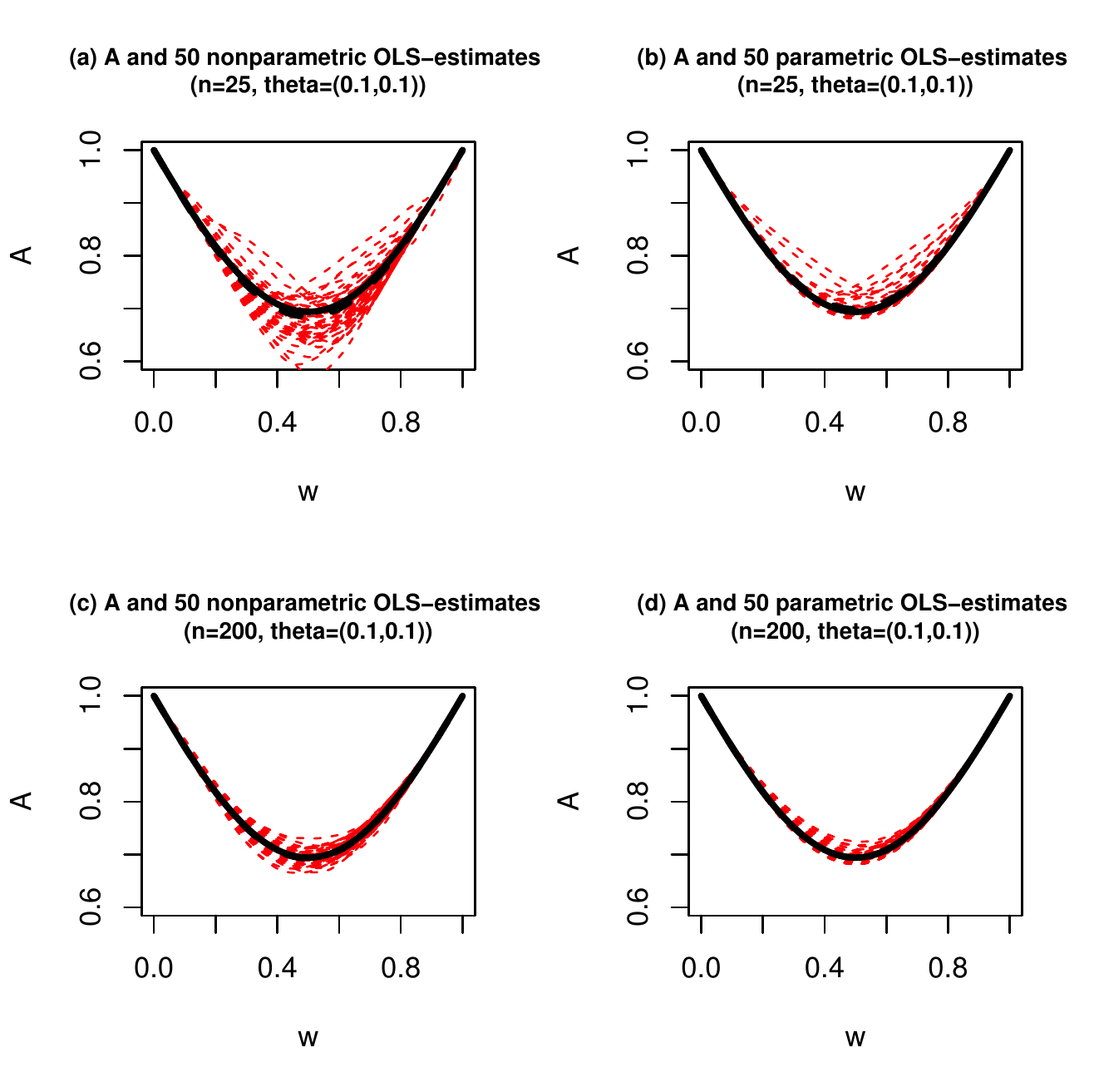}
	\caption{50 estimates $\hat{A}_{OLS}(\mathbf{w})$ (fig. (a) and (c))
	and $A\left( \mathbf{w},\mathbf{\hat{\theta}}_{M,OLS}\right)$ 
	for $\theta=(0.1,0.1)$ and $n\in\{25,\ 200\}$. The black line represents 
	the true function $A$.}
	\label{fig:A_50_estimates_theta0101}
\end{figure}

\begin{figure}[htbp]
	\centering
\includegraphics[height=0.5\textheight,width=0.9\textwidth]{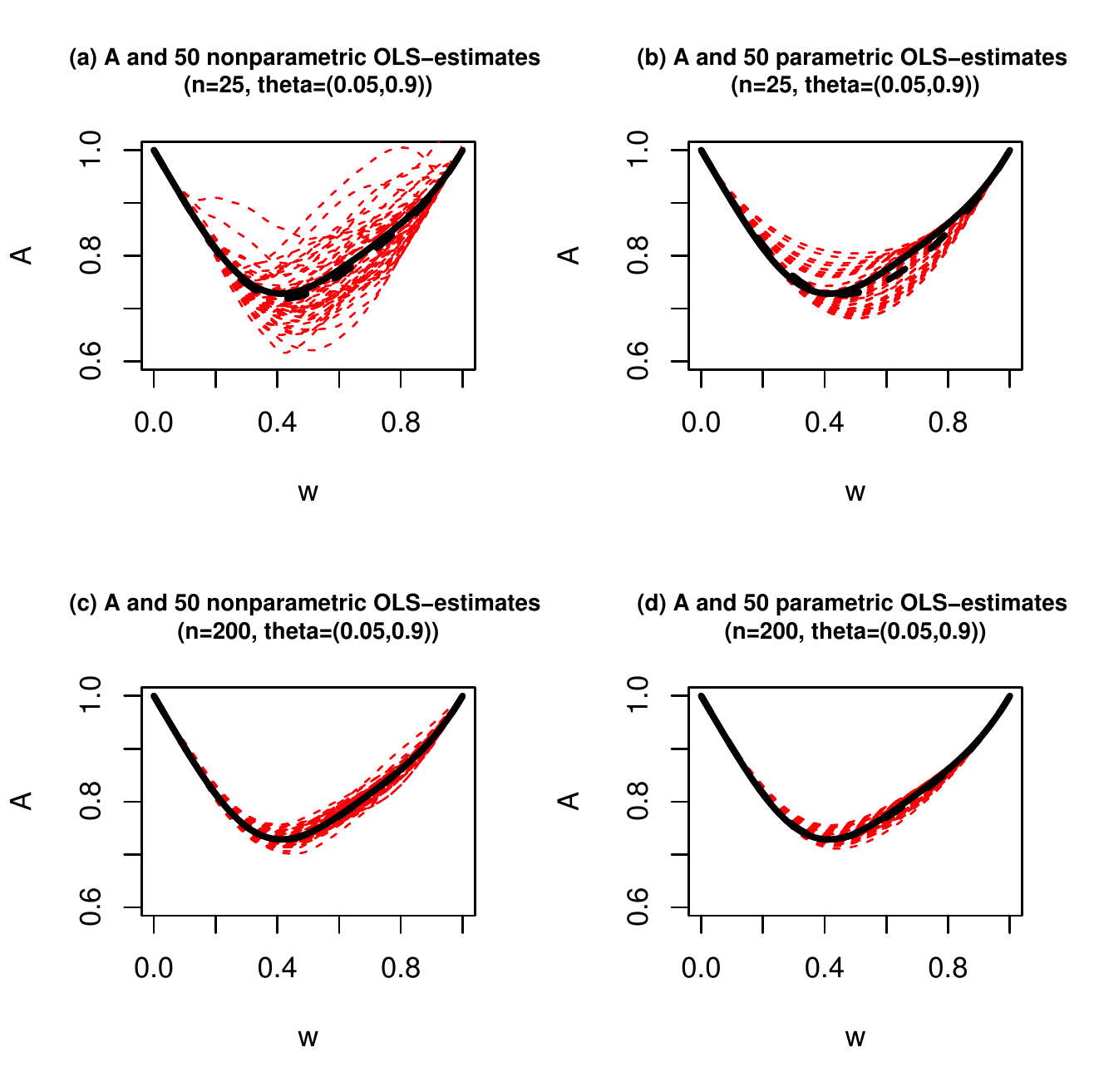}
	\caption{50 estimates $\hat{A}_{OLS}(\mathbf{w})$ (fig. (a) and (c))
	and $A(\mathbf{w},\mathbf{\hat{\theta}}_{M,OLS})$ 
	for $\theta=(0.05,0.9)$ and $n\in\{25,\ 200\}$. The black line represents 
	the true function $A$.}
	\label{fig:A_50_estimates_theta00509}
\end{figure}

\begin{figure}[htbp]
	\centering
\includegraphics[height=0.5\textheight,width=0.9\textwidth]
{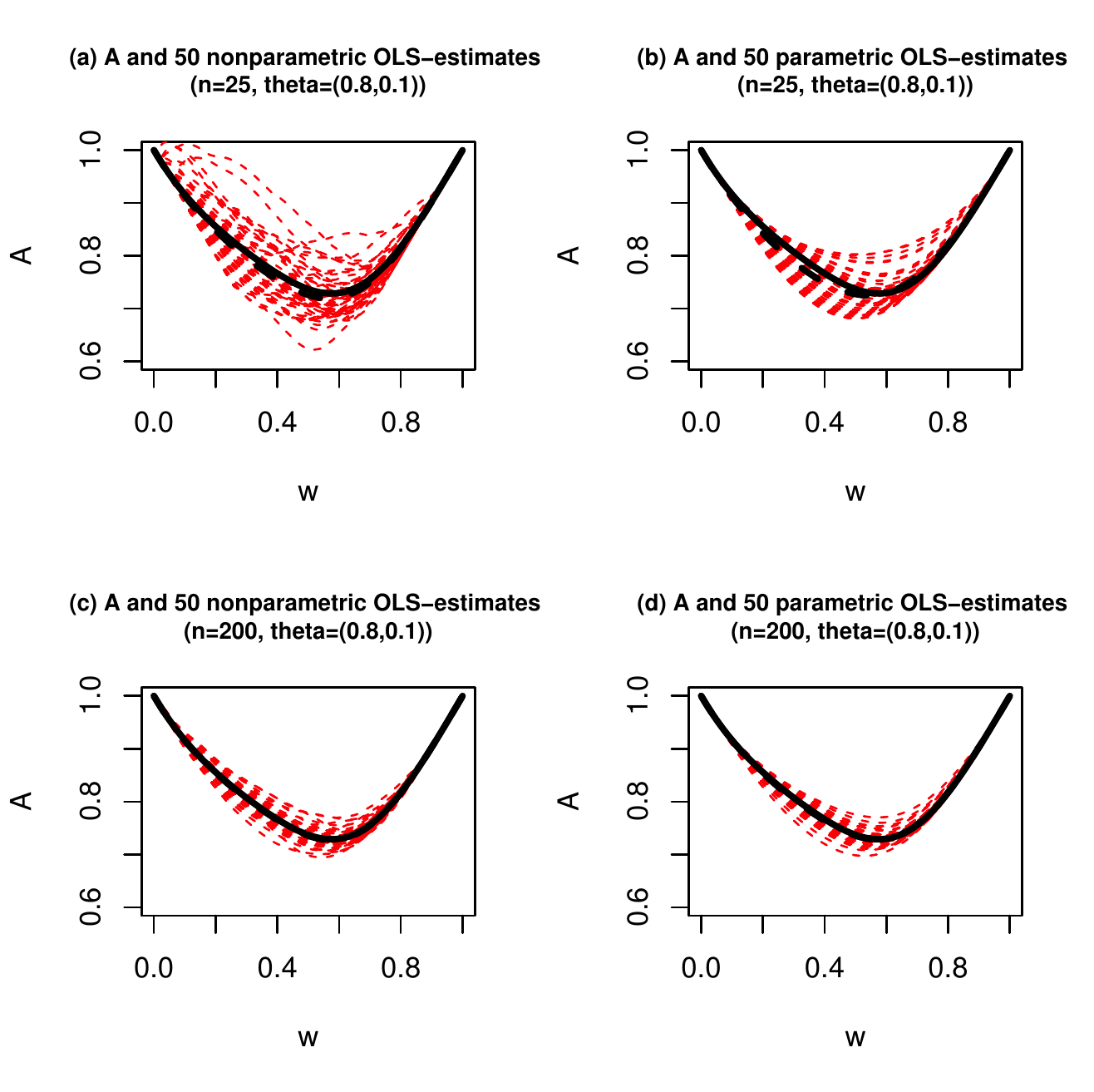}
	\caption{50 estimates $\hat{A}_{OLS}(\mathbf{w})$ (fig. (a) and (c))
	and $A( \mathbf{w},\mathbf{\hat{\theta}}_{M,OLS})$ 
	for $\theta=(0.8,0.1)$ and $n\in\{25,\ 200\}$. The black line represents 
	the true function $A$.}
	\label{fig:A_50_estimates_theta0801}
\end{figure}

\begin{figure}[htbp]
	\centering
\includegraphics[height=0.3\textheight,width=0.3\textwidth]{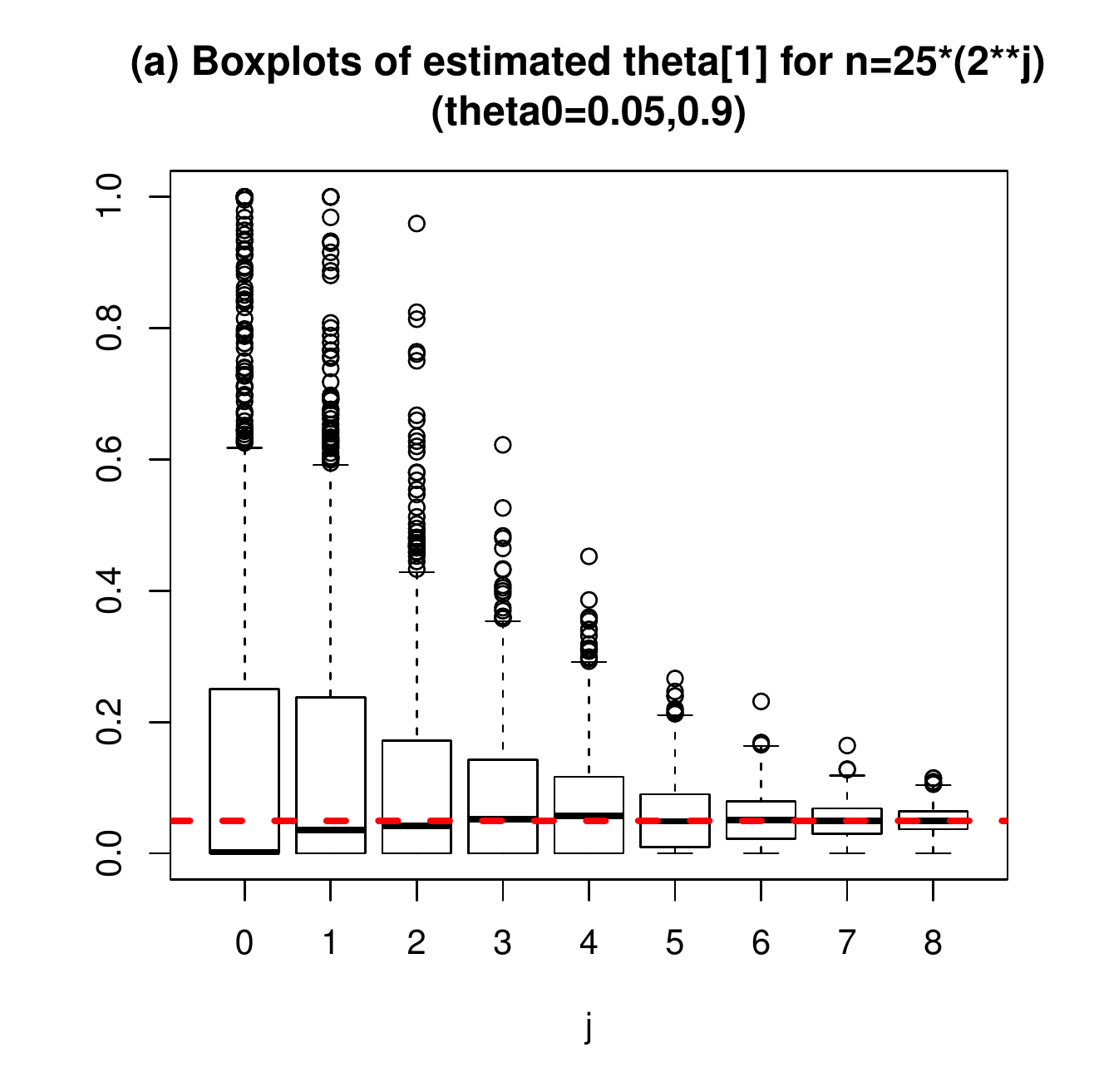}
\includegraphics[height=0.3\textheight,width=0.3\textwidth]{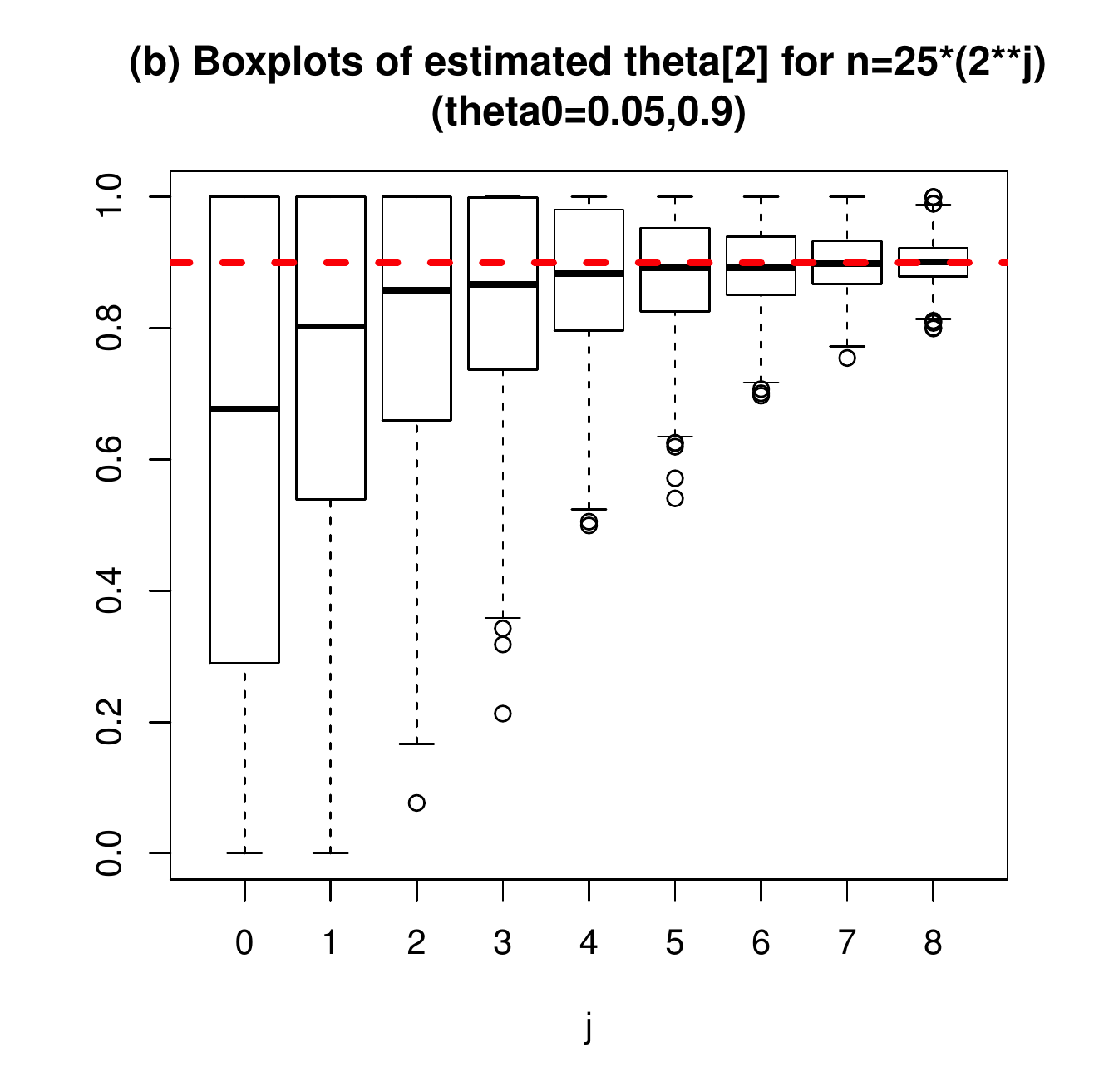}
	\caption{Boxplots of $\hat{\theta}_{OLS,1}$ (fig. a) 
	and $\hat{\theta}_{OLS,2}$ (fig. b)
	for the case with $\mathbf{\theta}=(0.05,0.9)$, and 
	$n=25\cdot 2^j$ $(j=0,1,...,8)$. The horizontal line
  represents the true value of $\theta_i$ ($i=1,2$).}
	\label{fig:box}
\end{figure}

\vfill\eject

\end{document}